\documentclass{article}

\PassOptionsToPackage{numbers}{natbib} 

\usepackage[final]{neurips_2019} 

\usepackage[utf8]{inputenc} 
\usepackage[T1]{fontenc}    
\usepackage{hyperref}       
\usepackage{url}            
\usepackage{booktabs}       
\usepackage{amsfonts}       
\usepackage{nicefrac}       
\usepackage{microtype}      
\usepackage{wrapfig}

\usepackage{amsthm}
\usepackage{amsmath}
\usepackage{amsfonts}
\usepackage{amssymb}
\usepackage{stmaryrd}
\usepackage{bm}
\usepackage[pdftex]{graphicx} 

\DeclareMathOperator*{\argmax}{arg\,max}

\DeclareMathOperator*{\rank}{rank}
\DeclareMathOperator*{\diag}{diag}
\DeclareMathOperator*{\vect}{vect}

\newcommand{\bbR}{\mathbb{R}} 

\newcommand{\bbE}{\mathbb{E}} 
\newcommand{\VAR}{\mathbb{V}\text{ar}} 
\newcommand{\COV}{\mathbb{C}\text{ov}} 
\newcommand{\bXi}{\bm{X}_i} 
\newcommand{\bCi}{\bm{C}_i} 
 
\newcommand{\cM}{\mathcal{M}} 
\newcommand{\Sp}{\mathcal{S}_P} 
\newcommand{\SpP}{\mathcal{S}_P^+} 
\newcommand{\SpPP}{\mathcal{S}_P^{++}}

\newcommand{\SrPP}{\mathcal{S}_R^{++}}
\newcommand{\SprP}{\mathcal{S}_{P, R}^+}

\usepackage[textwidth=\marginparwidth, textsize=tiny]{todonotes}
\newcommand{\iid}{{i.i.d.~}}
\newcommand{\ie}{{\em i.e.~}}

\newcommand{\wrt}{{\em w.r.t.~}}

\newtheorem{prop}{Proposition}

\title{Manifold-regression to predict from MEG/EEG brain signals without source
modeling}

\author{
	David Sabbagh 
	\thanks{Additional affiliation: Inserm, UMRS-942, Paris Diderot University, Paris, France}
	\thanks{Additional affiliation: Department of Anaesthesiology and Critical Care, Lariboisi\`ere
	 Hospital, Assistance Publique H\^opitaux de Paris, Paris, France} 
	\thanks{\texttt{dav.sabbagh@gmail.com}}
	, Pierre Ablin, Ga\"el Varoquaux, Alexandre Gramfort,
	Denis A. Engemann 
	\thanks{\texttt{denis-alexander.engemann@inria.fr}} \\
Université Paris-Saclay, Inria, CEA, Palaiseau, 91120, France \\
}

\begin{document}

\bibliographystyle{plain} 

\maketitle

\begin{abstract} 

Magnetoencephalography and electroencephalography (M/EEG) can reveal neuronal dynamics non-invasively in real-time and are therefore appreciated methods in medicine and neuroscience.
Recent advances in modeling brain-behavior relationships have highlighted the effectiveness of Riemannian geometry for summarizing the spatially correlated time-series from M/EEG in terms of their covariance. However, after artefact-suppression, M/EEG data is often rank deficient which limits the application of Riemannian concepts.
In this article, we focus on the task of \emph{regression} with rank-reduced covariance matrices. We study two Riemannian approaches that vectorize the M/EEG covariance between-sensors through projection into a tangent space. The Wasserstein distance readily applies to rank-reduced data but lacks affine-invariance.
This can be overcome by finding a common subspace in which the covariance matrices are full rank, enabling the affine-invariant geometric distance.
We investigated the implications of these two approaches in synthetic generative models, which allowed us to control estimation bias of a linear model for prediction.
We show that Wasserstein and geometric distances allow perfect out-of-sample prediction on the generative models.
We then evaluated the methods on real data with regard to their effectiveness in predicting age from M/EEG covariance matrices.
The findings suggest that the data-driven Riemannian methods outperform different sensor-space estimators and that they get close to the performance of biophysics-driven source-localization model that requires MRI acquisitions and tedious data processing.
Our study suggests that the proposed Riemannian methods can serve as fundamental building-blocks for automated large-scale analysis of M/EEG.

\end{abstract}

\section{Introduction}

Magnetoencephalography and electroencephalography (M/EEG) measure brain activity with millisecond precision from outside the head~\cite{hamalainen1993magnetoencephalography}. Both methods are non-invasive and expose rhythmic signals induced by coordinated neuronal firing with characteristic periodicity between minutes and milliseconds~\cite{Buzsaki:2017}. These so-called brain-rhythms can reveal cognitive processes as well as health status and are quantified in terms of the spatial distribution of the power spectrum over the sensor array that samples the electromagnetic fields around the head~\cite{baillet2017}.

Statistical learning from M/EEG commonly relies on covariance matrices estimated from band-pass filtered signals to capture the characteristic scale of the neuronal events of interest~\cite{blankertz-etal:08,grosse-wentrup:08,dmochowski2012correlated}. However, covariance matrices do not live in an Euclidean space but a Riemannian manifold. Fortunately, Riemannian geometry offers a principled mathematical approach to use standard linear learning algorithms such as logistic or ridge regression that work with Euclidean geometry. This is achieved by projecting the covariance matrices into a vector space equipped with an Euclidean metric, the tangent space.
The projection is defined by the Riemannian metric, for example the geometric affine-invariant metric~\cite{Bhatia:07} or the Wasserstein metric~\cite{bhatia2018bures}. As a result, the prediction error can be substantially reduced when learning from covariance matrices using Riemannian methods~\cite{yger-etal:17,congedo-etal:17}.

In  practice, M/EEG data is often provided in a rank deficient form by platform operators but also curators of public datasets~\cite{larson2013adding,lemon}. Its contamination with high-amplitude environmental electromagnetic artefacts often render aggressive offline-processing mandatory to yield intelligible signals. Commonly used tools for artefact-suppression project the signal linearly into a lower dimensional subspace that is hoped to predominantly contain brain
signals~\cite{taulu,uusitalo1997signal,makeig-etal:1995}. But this necessarily leads to inherently rank-deficient covariance matrices for which no affine-invariant distance is defined. One remedy may consist in using anatomically informed source localization techniques that can typically deal with rank deficiencies~\cite{engemann2015automated} and can be combined with source-level estimators of neuronal interactions~\cite{khan2018maturation}. However, such approaches require domain-specific expert knowledge, imply processing steps that are hard to automate (e.g. anatomical coregistration) and yields pipelines in which excessive amounts of preprocessing are not under control of the predictive model.

In this work, we focus on regression with rank-reduced covariance matrices. We propose two Riemannian methods for this problem. A first approach uses a Wasserstein metric that can handle rank-reduced matrices, yet is not affine-invariant. In a second approach, matrices are projected into a common subspace in which affine-invariance can be provided.
We show that both metrics can achieve perfect out-of-sample predictions in a synthetic generative model.
Based on the SPoC method \cite{dahne2014spoc}, we then present a supervised and computationally efficient approach to learn subspace projections informed by the target variable. Finally, we apply these models to the problem of inferring age from brain data~\cite{liem-etal:2017,khan2018maturation} on 595 MEG recordings from the Cambridge Center of Aging (Cam-CAN, http://cam-can.org) covering an age range from 18 to 88 years \cite{taylor2017cambridge}. We compare the data-driven Riemannian approaches to simpler methods that extract power estimates from the diagonal of the sensor-level covariance as well as the cortically constrained minimum norm estimates (MNE) which we use to project the covariance into a subspace defined by anatomical prior knowledge.

\paragraph{Notations}
We denote scalars $s \in \bbR$ with regular lowercase font, vectors $\bm{s} =
[s_{1}, \ldots, s_{N}] \in \bbR^{N}$ with bold lowercase font and matrices
$\bm{S} \in \bbR^{N \times M}$ with bold uppercase fonts.
$\bm{I}_N$ is the identity matrix of size $N$. $[\cdot]^{\top}$ represents vector or matrix transposition.
The Frobenius norm of a matrix will be denoted by $||\bm{M}||_{F}^{2} = \text{Tr}(\bm{M}\bm{M}^{\top}) = \sum |M_{ij}|^{2}$ with $\text{Tr}(\cdot)$ the trace operator. $\rank(\bm{M})$ is the rank of a matrix.
The $l_{2}$ norm of a vector $\bm{x}$ is denoted by $||\bm{x}||_{2}^{2} = \sum x_{i}^{2}$.
We denote by $\cM_P$ the space of $P \times P$ square real-valued matrices, 
$\Sp =\{\bm{M} \in \cM_{P}, \bm{M}^{\top} = \bm{M} \}$ the subspace of symmetric
matrices,
$\SpPP = \{ \bm{S} \in \Sp, \bm{x}^{\top} S \bm{x} > 0, \forall \bm{x} \in
\bbR^{P} \}$ the subspace of $P \times P$ symmetric positive definite matrices,
$\SpP = \{ \bm{S} \in \Sp, \bm{x}^{\top} S \bm{x} \geq 0, \forall \bm{x} \in
\bbR^{P} \}$ the subspace of $P \times P$ symmetric semi-definite positive (SPD)
matrices,
$\SprP = \{ \bm{S} \in \SpP, \rank(\bm{S}) = R\}$ the subspace of SPD matrices of fixed
rank R.
All matrices $\bm{S} \in \SpPP$ are full rank, invertible (with $\bm{S}^{-1} \in \SpPP$) and
diagonalizable with real strictly positive eigenvalues: $\bm{S} = \bm{U}
\bm{\Lambda} \bm{U}^{\top}$ with $\bm{U}$ an orthogonal matrix of eigenvectors
of $\bm{S}$ ($\bm{U} \bm{U}^{\top} = \bm{I}_P$) and $\bm{\Lambda} = \diag(\lambda_{1}, \ldots, \lambda_{n})$
the diagonal matrix of its eigenvalues $\lambda_{1} \geq \ldots \geq \lambda_{n} >
0$.  For a matrix $\bm{M}$, $\diag(\bm{M}) \in \bbR^P$ is its diagonal.
We also define the exponential and logarithm of a matrix: 
$\forall \bm{S} \in \SpPP, \log(\bm{S}) =
\bm{U}~\diag(\log(\lambda_{1}), \ldots, \log(\lambda_{n})) ~\bm{U}^{\top}
\in \Sp$,
and $\forall \bm{M} \in \Sp, \exp(\bm{M}) =\bm{U}~\diag(\exp(\lambda_{1}), \ldots, \exp(\lambda_{n}))
~\bm{U}^{\top} \in \SpPP$.
$\mathcal{N}(\mu, \sigma^2)$ denotes the normal (Gaussian) distribution of mean $\mu$ and variance $\sigma^2$. Finally, $\bbE_s[\bm{x}]$ represents the expectation and $\VAR_s[\bm{x}]$ the
variance of any random variable $\bm{x}$ \wrt their subscript $s$ when needed.

\paragraph{Background and M/EEG generative model}

MEG or EEG data measured on $P$ channels are multivariate signals $\bm{x}(t) \in \bbR^{P}$.
For each subject $i=1 \dots N$, the data are
a matrix $\bXi \in \bbR^{P \times T}$ where $T$ is the number of time samples.
For the sake of simplicity, we assume that $T$ is the same for each subject,
although it is not required by the following method.
The \emph{linear instantaneous mixing model} is a valid generative model for
M/EEG data due to the linearity of Maxwell's equations~\cite{hamalainen1993magnetoencephalography}.
Assuming the signal originates from $Q<P$ locations in the brain, at any time $t$, the measured signal vector of subject $i=1 \dots N$ is a linear combination of the $Q$ \emph{source patterns} $\bm{a}^s_j \in \bbR^P$, $j=1 \dots Q$:
\begin{equation}
    \label{eq:generativemodel}
    \bm{x}_i(t) = \bm{A}^s~\bm{s}_i(t) + \bm{n}_i(t) \enspace ,
\end{equation}
where the patterns form the time and subject-independent
source \emph{mixing matrix} $\bm{A}^s = [\bm{a}^s_{1}, \ldots, \bm{a}^s_{Q}] \in \bbR^{P \times Q}$, $\bm{s}_i(t) \in \bbR^{Q}$ is the \emph{source vector} formed by the $Q$ time-dependent sources amplitude, $\bm{n}_i(t) \in \bbR^{P}$ is a contamination due to noise. Note that the mixing matrix $\bm{A}^s$ and sources $\bm{s}_i$ are not known.

Following numerous learning models on M/EEG~\cite{blankertz-etal:08,dahne2014spoc,grosse-wentrup:08}, we consider a regression setting where the target $y_i$ is a function of the power of the sources, denoted $p_{i, j}=\bbE_t[s_{i, j}^2(t)]$. Here we consider the linear model:
\begin{equation}
    \label{eq:regmodel}
    y_i = \sum_{j=1}^Q \alpha_j f(p_{i,j}) \enspace ,
\end{equation}
where $\bm{\alpha}\in \bbR^Q$ and  $f: \bbR^+ \to \bbR$ is increasing.
Possible choices for $f$ that are relevant for neuroscience are
$f(x)=x$, or $f(x)=\log(x)$ to account for
log-linear relationships between brain signal power and cognition
~\citep{blankertz-etal:08,grosse-wentrup:08,buzsaki2014log}.
A first approach consists in estimating the sources before fitting such a linear model, for example using the Minimum Norm Estimator (MNE) approach~\cite{Hamalainen:1984}. This boils down to solving the so-called M/EEG inverse problem which requires costly MRI acquisitions and tedious processing~\cite{baillet2017}. A second approach is to work directly with the signals $\bm{X}_i$. To do so, models that enjoy some invariance property are desirable: these models are blind to the mixing $\bm{A}^s$ and working with the signals $\bm{x}$ is similar to working directly with the sources $\bm{s}$.
Riemannian geometry is a natural setting where such invariance properties are found~\cite{forstner2003metric}.
Besides, under Gaussian assumptions, model \eqref{eq:generativemodel} is fully
described by second order statistics~\citep{rodrigues2018multivariate}. This
amounts to working with covariance matrices, $\bCi = \bXi \bXi^{\top} / T$, for
which Riemannian geometry is well developed. One specificity of M/EEG data is,
however, that signals used for learning have been rank-reduced. This leads to rank-deficient covariance matrices, $\bCi \in \SprP$, for which specific matrix manifolds need to be considered.

\section{Theoretical background to model invariances on $\SprP$ manifold}
\label{sec:manifold}
\subsection{Riemannian matrix manifolds}
\begin{wrapfigure}{r}{0.4\textwidth}
\vspace{-0.4cm}
\hfill
\includegraphics[width=\linewidth]{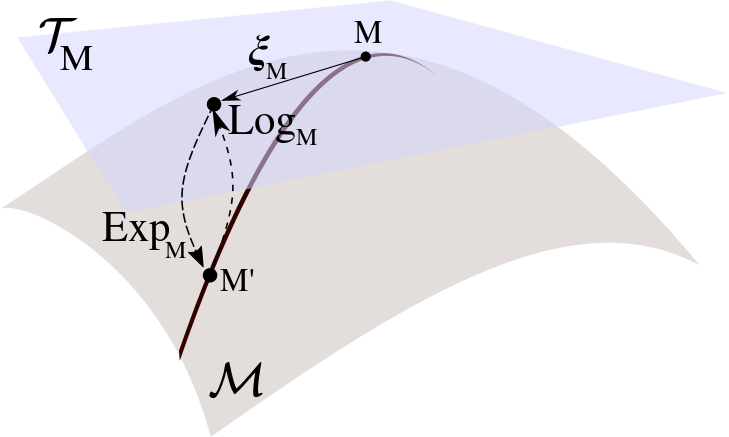}
\caption{Tangent Space, exponential and logarithm on Riemannian manifold illustration.}%
\label{fig:manifold}%
\end{wrapfigure}

Endowing a continuous set $\cM$ of square matrices with a metric, that defines a
local Euclidean structure, gives a Riemannian manifold with a solid theoretical framework.
Let $\bm{M}\in \cM$, a $K$-dimensional Riemannian manifold. For any matrix $\bm{M}' \in \cM$, as $\bm{M}' \to \bm{M}$, $ \bm{\xi_M} = \bm{M}' - \bm{M}$ belongs to a vector space
$\mathcal{T}_{\bm{M}}$ of dimension $K$ called the \emph{tangent space} at $\bm{M}$.

The \emph{Riemannian metric} defines an inner product $\langle \cdot ,\cdot \rangle_{\bm{M}}: \enspace \mathcal{T}_{\bm{M}} \times \mathcal{T}_{\bm{M}} \to \bbR$ for each tangent space $\mathcal{T}_{\bm{M}}$, and as a consequence a norm in the tangent space $\|\bm{\xi}\|_{\bm{M}}= \sqrt{\langle \bm{\xi} ,\bm{\xi} \rangle_{\bm{M}}}$.
Integrating this metric between two points gives a \emph{geodesic} distance $d: \cM \times \cM \rightarrow \bbR^+$.
It allows to define means on the manifold:
\begin{equation}
    \label{eq:average}
    \text{Mean}_d(M_1,\dots, M_N) = \arg\min_{M \in \cM} \sum_{i=1}^N d(M_i, M)^2 \enspace .
\end{equation}

The \emph{manifold exponential} at $\bm{M} \in \cM$, denoted $\text{Exp}_{\bm{M}}$, is a smooth mapping from $T_M$ to $\cM$ that preserves local properties. In particular,
$d(\text{Exp}_{\bm{M}}(\bm{\xi_{\bm{M}}}), \bm{M}) = \|\bm{\xi_{\bm{M}}}\|_{\bm{M}} +
o(\|\bm{\xi_{\bm{M}}}\|_{\bm{M}})$.
Its inverse is the \emph{manifold logarithm} $\text{Log}_{\bm{M}}$ from $\cM$ to $\mathcal{T}_{\bm{M}}$, with $\|\text{Log}_{\bm{M}}(\bm{M}')\|_{\bm{M}} = d(\bm{M}, \bm{M}') + o(d(\bm{M}, \bm{M}'))$ for $\bm{M}, \bm{M}' \in
\cM$.
Finally, since $\mathcal{T}_{\bm{M}}$ is Euclidean, there is a linear invertible mapping $ \phi_{\bm{M}}: \mathcal{T}_{\bm{M}} \to \bbR^K$ such that for all $\xi_{\bm{M}} \in \mathcal{T}_{\bm{M}}$, $\|\bm{\xi}_{\bm{M}}\|_{\bm{M}} = \|\phi_{\bm{M}}(\bm{\xi}_{\bm{M}})\|_2$.
This allows to define the \emph{vectorization operator} at $\bm{M} \in \cM$, $\mathcal{P}_{\bm{M}}: \cM \to \mathbb{R}^K$, defined by $ \mathcal{P}_{\bm{M}}(\bm{M'}) = \phi_{\bm{M}}(\text{Log}_{\bm{M}}(\bm{M}'))$.
Fig. \ref{fig:manifold} illustrates these concepts.

The vectorization explicitly captures the local Euclidean properties of the Riemannian manifold:
\begin{equation}
    \label{eq:isometry}
    d(\bm{M}, \bm{M'}) = \|\mathcal{P}_{\bm{M}}(\bm{M'})\|_2 + o(\|\mathcal{P}_{\bm{M}}(\bm{M'})\|_2)
\end{equation}
Hence, if a set of matrices $\bm{M}_1, \dots, \bm{M}_N$ is located in a small portion of the manifold, denoting $\overline{\bm{M}}= \text{Mean}_d(\bm{M}_1, \dots, \bm{M}_N)$, it holds:%
\begin{equation}
    \label{eq:distance_preservation}
    d(\bm{M}_i, \bm{M}_j) \simeq \|\mathcal{P}_{\overline{\bm{M}}}(\bm{M}_i) - \mathcal{P}_{\overline{\bm{M}}}(\bm{M}_j)\|_2
\end{equation}

For additional details on matrix manifolds, see ~\citep{absil2009optimization}, chap. 3.

\paragraph{Regression on matrix manifolds}

The vectorization operator is key for machine learning applications: it projects
points in $\cM$ on $\bbR^K$, and the distance $d$ on $\cM$ is approximated
by the distance $\ell_2$ on $\bbR^K$.
Therefore, those vectors can be used as input for any standard regression technique, which often assumes a Euclidean structure of the data.
More specifically, throughout the article, we consider the following regression pipeline. Given a training set of samples $\bm{M}_1, \dots, \bm{M}_N\in \cM$ and
target continuous variables $y_1,\dots, y_N \in \bbR$, we first compute the mean of the samples $\overline{\bm{M}}= \text{Mean}_d(\bm{M}_1, \dots, \bm{M}_N)$.
This mean is taken as the reference to compute the vectorization. After computing $\bm{v}_1, \dots, \bm{v}_N \in
\bbR^K$ as $\bm{v}_i=\mathcal{P}_{\overline{\bm{M}}}(\bm{M}_i)$, a linear regression technique (e.g. ridge regression) with parameters $\bm{\beta} \in \bbR^K$ can be employed assuming that $y_i \simeq \bm{v}_i^{\top}\bm{\beta}$.

\subsection{Distances and invariances on positive matrices manifolds}
We will now introduce two important distances: the geometric distance
on the manifold $\SpPP$ (also known as affine-invariant distance), and the Wasserstein distance on the manifold $\SprP$.
\paragraph{The geometric distance}
Seeking properties of covariance matrices that are invariant by linear transformation of the signal, leads to endow  
the positive definite manifold $S_P^{++}$ with the \emph{geometric} distance
\cite{forstner2003metric}:
\begin{equation} \label{eq:geom_dist}
d_G(\bm{S}, \bm{S'}) = \|\log(\bm{S}^{-1} \bm{S'}) \|_{F} =
\Bigg[\sum_{i=1}^{P} \log^{2} \lambda_{k} \Bigg]^{\frac12}
\end{equation}
where $\lambda_{k}, k=1\ldots P$ are the real eigenvalues of $\bm{S}^{-1} \bm{S'}$.
The affine invariance property writes:
\begin{equation} \label{eq:invariance}
\text{For } \bm{W} \text{ invertible, }
d_{G}(\bm{W}^{\top}
\bm{S} \bm{W}, \bm{W}^{\top} \bm{S}' \bm{W}) = d_{G}(\bm{S},
\bm{S}') \enspace.
\end{equation}
This distance gives a Riemannian-manifold structure to  $S_P^{++}$ with the inner product $ \langle \bm{P},\bm{Q} \rangle_{\bm{S}} = \text{Tr}(\bm{P}\bm{S}^{-1}\bm{Q}\bm{S}^{-1})$
\cite{forstner2003metric}.
The corresponding manifold logarithm at $\bm{S}$ is $\text{Log}_{\bm{S}}(\bm{S'}) = \bm{S}^{\frac12}
\log\big( \bm{S}^{-\frac12}\bm{S'} \bm{S}^{-\frac12}\big) \bm{S}^{\frac12}$
and the vectorization operator $\mathcal{P}_{\bm{S}}(\bm{S'})$ of $\bm{S'}$
\wrt $\bm{S}$: $\mathcal{P}_{\bm{S}}(\bm{S'}) =
\text{Upper}(\bm{S}^{-\frac12}\text{Log}_{\bm{S}}(\bm{S'})\bm{S}^{-\frac 12}) =
\text{Upper}(\log(\bm{S}^{-\frac12}\bm{S'}\bm{S}^{-\frac 12}))$, where
$\text{Upper}(\bm{M})\in \bbR^K$ is the vectorized upper-triangular part of
$\bm{M}$, with unit weights on the diagonal and $\sqrt{2}$ weights on the off-diagonal, and $K=P(P+1)/2$.

\paragraph{The Wasserstein distance} Unlike $\SpPP$, it is hard to endow
the $\SprP$ manifold with a distance that yields tractable or cheap-to-compute
logarithms~\citep{vandereycken2009embedded}.  This manifold is classically
viewed as $\SprP = \{\mathbf{Y}\mathbf{Y}^{\top}|\mathbf{Y} \in \bbR_*^{P \times R}\}$, where $\bbR_*^{P
\times R}$ is the set $P \times R$ matrices of rank $R$~\citep{journee2010low}.
This view allows to write $\SprP$ as a quotient manifold $\bbR_*^{P \times R} /
\mathcal{O}_R$, where $\mathcal{O}_R$ is the orthogonal group of size $R$.
This means that each matrix $\mathbf{Y}\mathbf{Y}^{\top} \in \SprP$ is identified with the set
$\{\mathbf{Y}\mathbf{Q}| \mathbf{Q} \in \mathcal{O}_R\}$.

It has recently been proposed~\citep{massart2018psd} to use the standard
Frobenius metric on the total space $\bbR_*^{P \times R}$.
This metric in the total space is equivalent to the \emph{Wasserstein} distance~\cite{bhatia2018bures} on $\SprP$:
\begin{equation}
  \label{eq:wasserstein}
    d_W(\bm{S}, \bm{S}') = \Big[\text{Tr}(\bm{S}) + \text{Tr}(\bm{S}') - 2
    \text{Tr}((\bm{S}^{\frac12}\bm{S}'\bm{S}^{\frac12})^\frac12)\Big]^\frac12
\end{equation}
This provides cheap-to-compute logarithms:
\begin{equation}
    \text{Log}_{\bm{Y}\bm{Y}^{\top}}(\bm{Y}'\bm{Y}'^{\top}) = \bm{Y}'\bm{Q}^* - \bm{Y} \in \bbR_*^{P \times R} \enspace,
\end{equation}
where $\bm{U}\bm{\Sigma}\bm{V}^{\top} = \bm{Y}^{\top} \bm{Y}'$ is a singular
value decomposition and $\bm{Q}^* = \bm{V}\bm{U}^{\top}$.
The vectorization operator is then given by $\mathcal{P}_{\bm{Y}\bm{Y}^{\top}}(\bm{Y}'\bm{Y}'^{\top}) = \vect(\bm{Y}'\bm{Q}^* - \bm{Y}) \in \bbR^{PR}$, where the $\vect$ of a matrix is the vector containing all its coefficients.

This framework offers closed form projections in the tangent space for
the Wasserstein distance, which can be used to perform regression.
Importantly, since $S_P^{++}=S_{P, P}^+$, we can also use this distance on the positive definite matrices. This distance possesses the \emph{orthogonal invariance} property:
\begin{equation} \label{eq:orth_invariance}
\text{For } \bm{W} \text{ orthogonal, }
d_{W}(\bm{W}^{\top}
\bm{S} \bm{W}, \bm{W}^{\top} \bm{S}' \bm{W}) = d_{W}(\bm{S},
\bm{S}') \enspace.
\end{equation}
This property is weaker than the affine invariance of the geometric
distance~(\ref{eq:invariance}).
A natural question is
whether such an affine invariant distance also exists on this manifold.
Unfortunately, it is shown in~\cite{bonnabel2009riemannian} that the answer is
negative for $R < P$ (proof in appendix \ref{subsec:noaff_inv_distance}).

\section{Manifold-regression models for M/EEG}

\subsection{Generative model and consistency of linear regression in the tangent space of $\SpPP$}
Here, we consider a more specific generative model than \eqref{eq:generativemodel} by assuming a specific structure on the noise. We assume that the additive noise $\bm{n}_i (t) = \bm{A}^n \bm{\nu}_i(t)$ with $\bm{A}^n = [\bm{a}^n_{1}, \ldots, \bm{a}^n_{P-Q}] \in \bbR^{P \times (P-Q)}$ and $\bm{\nu}_i(t) \in \bbR^{P-Q}$. This amounts to assuming that the noise is of rank $P-Q$ and that the noise spans the same subspace for all subjects.
Denoting $\bm{A} = [\bm{a}^s_{1}, \ldots, \bm{a}^s_{Q}, \bm{a}^n_{1}, \ldots, \bm{a}^n_{P-Q}] \in \bbR^{P \times P}$ and $\bm{\eta}_i(t) = [s_{i,1}(t), \dots s_{i, Q}(t), \nu_{i, 1}(t), \dots, \nu_{i, P-Q}(t)] \in \bbR^P$, this generative model can be compactly rewritten as $\bm{x}_i(t) = \bm{A}\bm{\eta}_i(t)$.

We assume that the sources $\bm{s}_i$ are decorrelated and independent from $\bm{\nu}_i$: with $p_{i, j}=\bbE_t[s_{i, j}^2(t)]$ the powers, \ie the variance over time, of the $j$-th source
of subject $i$, we suppose
$\bbE_t[\bm{s}_i(t)\bm{s}_i^{\top}(t)]= \diag((p_{i, j})_{j=1 \dots Q})$ and
$\bbE_t[\bm{s}_i(t)\bm{\nu}_i(t)^{\top}] = 0$. The covariances are then given by: 
\begin{equation}
\label{eq:covs}
    \bCi =\bm{A} \bm{E}_i\bm{A}^{\top} \enspace,
\end{equation}
where $\bm{E}_i = \bbE_t[\bm{\eta}_i(t)\bm{\eta}_i(t)^{\top}]$ is a 
block diagonal matrix, whose upper $Q\times Q$ block is $\diag(p_{i,1}, \dots, p_{i,Q})$.

%
In the following, we show that different functions $f$ from \eqref{eq:regmodel} yield a linear
relationship between the $y_i$'s and the vectorization of the $C_i$'s for
different Riemannian metrics.
\begin{prop}[Euclidean vectorization]
  \label{prop:consistency_eucl}
    Assume $f(p_{i,j}) = p_{i,j}$. 
    Then, the relationship between $y_i$ and $\text{Upper}(\bm{C}_i)$ is linear.
\end{prop}
\begin{proof}
Indeed, if $f(p) = p$, the relationship between $y_i$ and the $p_{i,j}$ is linear.
Rewriting Eq.~\eqref{eq:covs} as $\bm{E}_i=\bm{A}^{-1}\bm{C}_i\bm{A}^{-\top}$, and since the $p_{i, j}$ are on the diagonal of the upper block of $\bm{E}_i$, the relationship between the $p_{i, j}$ and the coefficients of $\bm{C}_i$ is also linear.
This means that there is a linear relationship between the coefficients of $\bm{C}_i$ and the variable of interest $y_i$.
In other words, $y_i$ is a linear combination of the vectorization of $\bm{C}_i$ \wrt the standard Euclidean distance.
\end{proof}
\begin{prop}[Geometric vectorization]
  \label{prop:consistency}
    Assume $f(p_{i,j}) = \log(p_{i,j})$. 
    Denote $\overline{\bm{C}} = \text{Mean}_G(\bm{C}_1,\dots, \bm{C}_N)$ the geometric mean of the dataset, and $ \bm{v}_i = \mathcal{P}_{\overline{\bm{C}}}(\bm{C}_i)$ the vectorization of $\bm{C}_i$ \wrt the geometric distance.
    Then, the relationship between $y_i$ and $\bm{v}_i$ is linear.
\end{prop}

The proof is given in appendix \ref{subsec:proof_prop}.  It relies crucially on the affine invariance
property that means that using Riemannian embeddings of the $\bm{C}_i$'s, is equivalent to working directly with the $\bm{E}_i$'s.

\begin{prop}[Wasserstein vectorization]
  \label{prop:consistency_wass}
    Assume $f(p_{i,j}) = \sqrt{p_{i,j}}$. 
    Assume that $\bm{A}$ is orthogonal.
    Denote $\overline{\bm{C}} = \text{Mean}_W(\bm{C}_1,\dots, \bm{C}_N)$ the Wasserstein mean of the dataset, and $ \bm{v}_i = \mathcal{P}_{\overline{\bm{C}}}(\bm{C}_i)$ the vectorization of $\bm{C}_i$ \wrt the Wasserstein distance.
    Then, the relationship between $y_i$ and $\bm{v}_i$ is linear.
\end{prop}

The proof is given in appendix \ref{subsec:proof_prop_wass}. The restriction to the case where $A$ is orthogonal stems from the orthogonal invariance of the Wasserstein distance.
In the neuroscience literature square root rectifications are however not commonly used for
M/EEG modeling. Nevertheless, it is interesting to see that the Wasserstein metric that can naturally cope with rank reduced data is consistent with this particular generative model.

These propositions show that the relationship between the samples and
the variable $y$ is linear in the tangent space, motivating the use of
linear regression methods (see simulation study in Sec.~\ref{sec:experiments}).
%
%
The argumentation of this section relies on the assumption that the covariance matrices are 
full rank. However, this is rarely the case in practice.

\subsection{Learning projections on $\SrPP$}
\label{sec:spatialfilter}



In order to use the geometric distance on the $\bm{C}_i \in \SprP$, we have to project them on $\SrPP$ to make them full rank.
In the following, we consider a linear operator $\bm{W} \in \bbR^{P \times R}$ of rank $R$ which is common to all samples (\ie subjects). For consistency with the M/EEG literature we will refer to rows of $\bm{W}$ as \emph{spatial filters}.
The covariance matrices of `spatially filtered' signals $\bm{W}^{\top}\bm{x}_i$ are obtained as: $ \bm{\Sigma}_{i} = \bm{W}^{\top} \bCi \bm{W} \in
\bbR^{R \times R}$.
With probability one, $\rank(\bm{\Sigma}_{i}) = \min(\rank(\bm{W}), \rank(\bCi)) = R$, hence $\bm{\Sigma}_{i} \in S_{R}^{++}$.
Since the $\bm{C}_i$'s do not span the same image, applying $\bm{W}$
 destroys some information.
Recently, geometry-aware dimensionality reduction techniques, both supervised and
unsupervised, have been developed
on covariance manifolds~\citep{horev2016geometry, harandi2017dimensionality}.
Here we considered two distinct approaches to estimate $\bm{W}$.
\paragraph{Unsupervised spatial filtering}
A first strategy is to project the data into a subspace that captures
most of its variance. This is achieved by Principal Component Analysis (PCA)
applied to the averaged covariance matrix computed across subjects: $ \bm{W}_{\text{UNSUP}} =
\bm{U}$, where $\bm{U}$ contains the eigenvectors corresponding
to the top $R$ eigenvalues of the average covariance matrix $\overline{\bm{C}}
= \frac{1}{N} \sum_{i=1}^{N} \bCi$. This step is blind to the values of $y$ and is therefore unsupervised.
Note that under the assumption that the time series across subjects are independent, the average covariance  $\overline{\bm{C}}$ is the covariance of the data over the full population.
\paragraph{Supervised spatial filtering}
We use a supervised spatial filtering algorithm \cite{dahne2014spoc} originally
developed for intra-subject Brain Computer Interfaces applications, and adapt it to our cross-person prediction problem.
The filters $\bm{W}$ are chosen to maximize the covariance between the power of
the filtered signals and $y$. Denoting by $\bm{C}_y = \frac{1}{N}
\sum_{i=1}^{N} y_{i} \bCi$ the weighted average covariance
matrix, the first filter $\bm{w}_{\text{SUP}}$ is given by:
$$
\bm{w}_{\text{SUP}} = \argmax_{\bm{w}} \frac{\bm{w}^{\top} \bm{C}_y
\bm{w}}{\bm{w}^{\top} \overline{\bm{C}} \bm{w}} \enspace .
$$
In practice, all the other filters in $\bm{W}_{\text{SUP}}$
are obtained by solving a generalized eigenvalue
decomposition problem (see the proof in Appendix ~\ref{subsec:appendix_spoc}).

The proposed pipeline is summarized in Fig.~\ref{fig:pipeline}.

\begin{figure}[tb]
	\begin{center}
        \includegraphics[width=\linewidth]{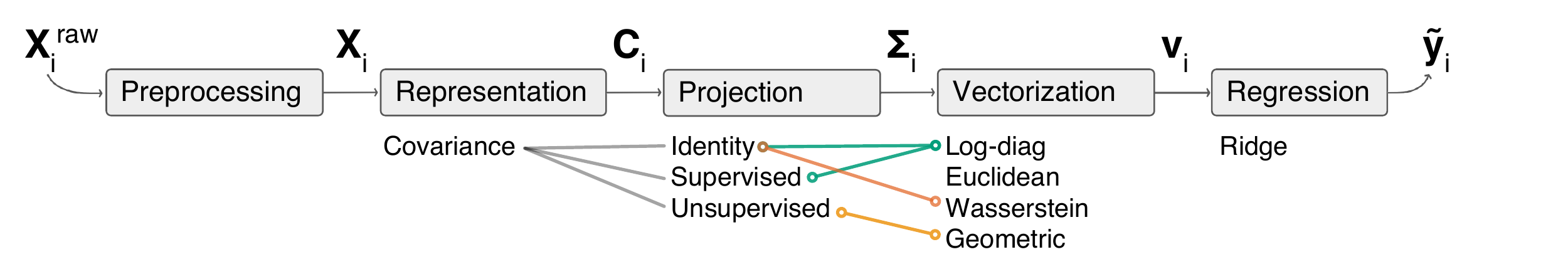}
	\end{center}
        \caption{Proposed regression pipeline. The considered choices
	for each sequential step are detailed below each box. Identity means no spatial filtering
	$\bm{W} = \bm{I}$. Only the most relevant combinations are reported. For example
	Wasserstein vectorization does not need projections as it
	directly applies to rank-deficient matrices. Geometric vectorization
    is not influenced by the choice of projections due to its affine-invariance. Choices for vectorization are depicted by the colors used for visualizing subsequent analyses.}
    \label{fig:pipeline}
\end{figure}




\section{Experiments}
\label{sec:experiments}

\subsection{Simulations}

We start by illustrating Prop.~\ref{prop:consistency}.
Independent identically distributed covariance matrices $\bm{C}_1, \dots, \bm{C}_N \in S_P^{++}$ and variables $y_1,
\dots, y_N$ are generated following the above generative model.
The matrix $\bm{A}$ is taken as $\exp(\mu \bm{B})$ with $\bm{B} \in \bbR^{P\times P}$ a random matrix, and $\mu \in \bbR$ a scalar controlling the distance from $A$ to identity ($\mu=0$ yields $\bm{A}= \bm{I}_P$).
We use the $\log$ function for $f$ to link the source powers (\ie the variance) to the $y_i$'s. Model reads $y_i = \sum_j \alpha_j \log(p_{ij}) + \varepsilon_i$, with $\varepsilon_i \sim \mathcal{N}(0, \sigma^2)$ a small additive random perturbation.

We compare three methods of vectorization: the geometric distance, the
Wasserstein distance and the non-Riemannian method ``log-diag'' extracting the $\log$ of the diagonals of $\bm{C}_i$ as features. Note that the diagonal of $\bm{C}_i$ contains the powers of each sensor for subject $i$.
%
%
A linear regression model is used following the procedure presented in Sec.~\ref{sec:manifold}.
We take $P=5$, $N=100$ and $Q = 2$.
We measure the score of each method as the average mean absolute error (MAE) obtained with 10-fold cross-validation.
Fig.~\ref{fig:synth_expe} displays the scores of each method when the parameters $\sigma$ controlling the noise level and $\mu$ controlling the distance from $A$ to $I_p$ are changed.
We also investigated the realistic scenario where each subject has a mixing matrix
deviating from a reference: $\bm{A}_i = \bm{A} + \bm{E}_i$ with entries of
$\bm{E}_i$ sampled \iid from $\mathcal{N}(0, \sigma^2)$.

The same experiment with $f(p) = \sqrt{p}$ yields comparable results, yet with Wasserstein distance performing best and achieving perfect out-of-sample prediction when $\sigma \to 0$ and $A$ is orthogonal.
\begin{figure}[h!]
    \centering
    \begin{minipage}{\linewidth}
        \begin{minipage}{0.34\linewidth}
            \includegraphics[width=1\linewidth]{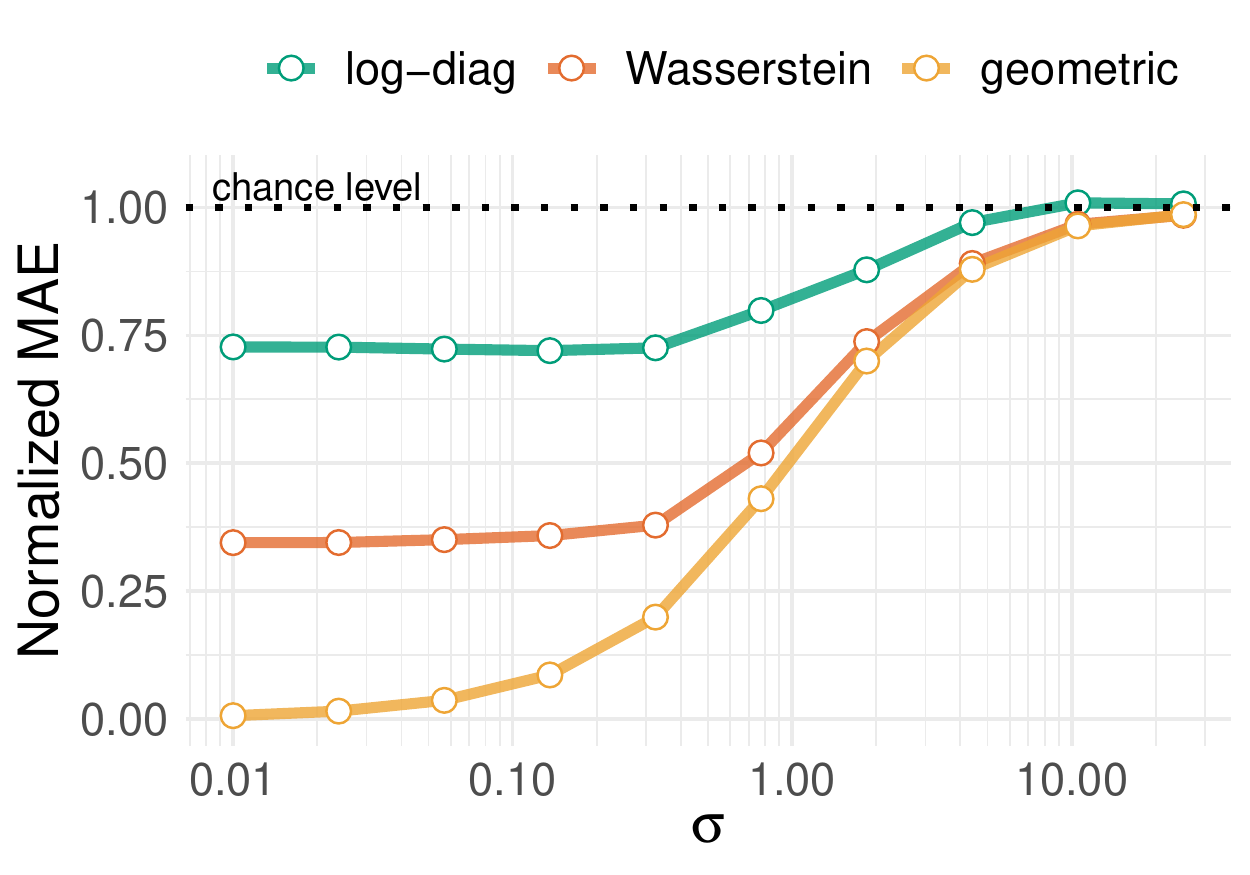}
        \end{minipage}
        \begin{minipage}{0.325\linewidth}
            \includegraphics[width=1\linewidth]{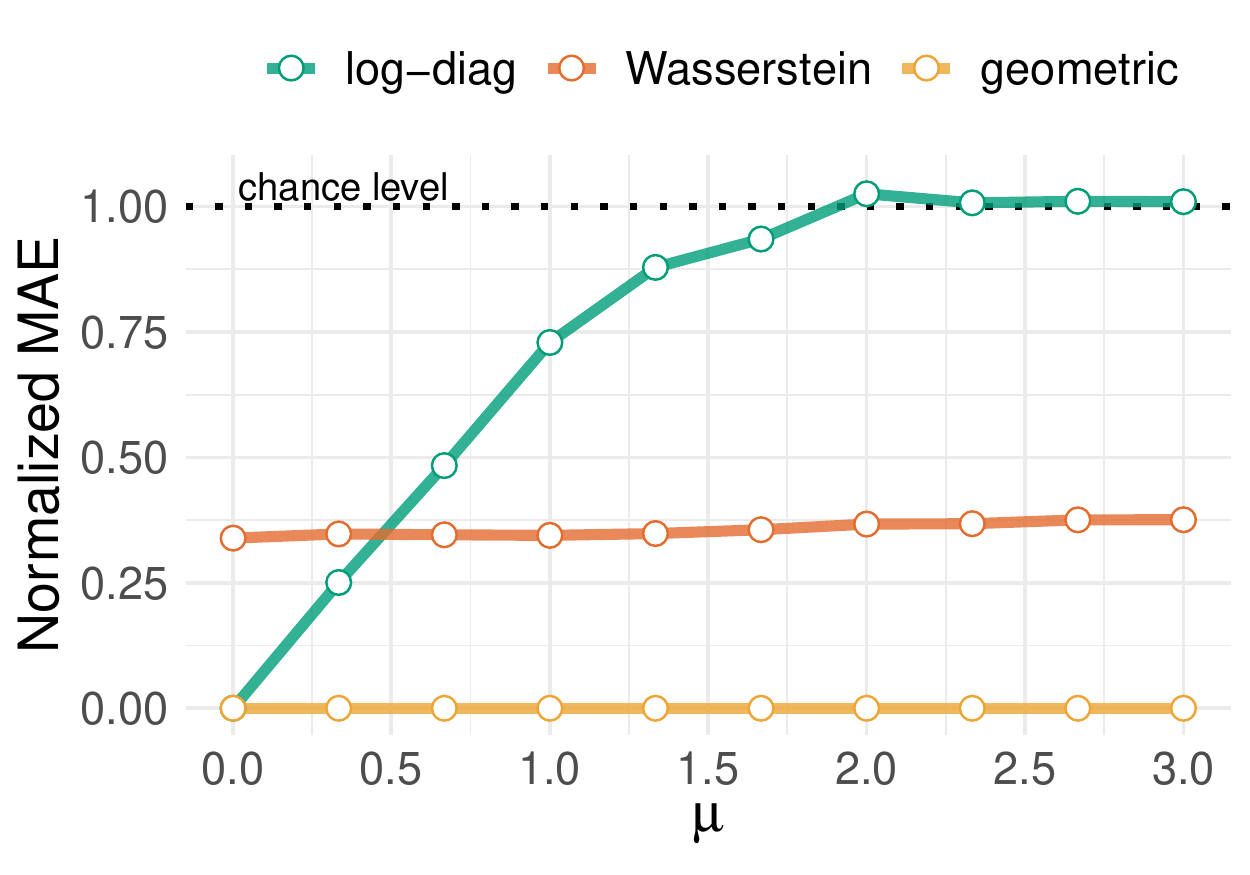}
        \end{minipage}
        \begin{minipage}{0.325\linewidth}
            \includegraphics[width=1\linewidth]{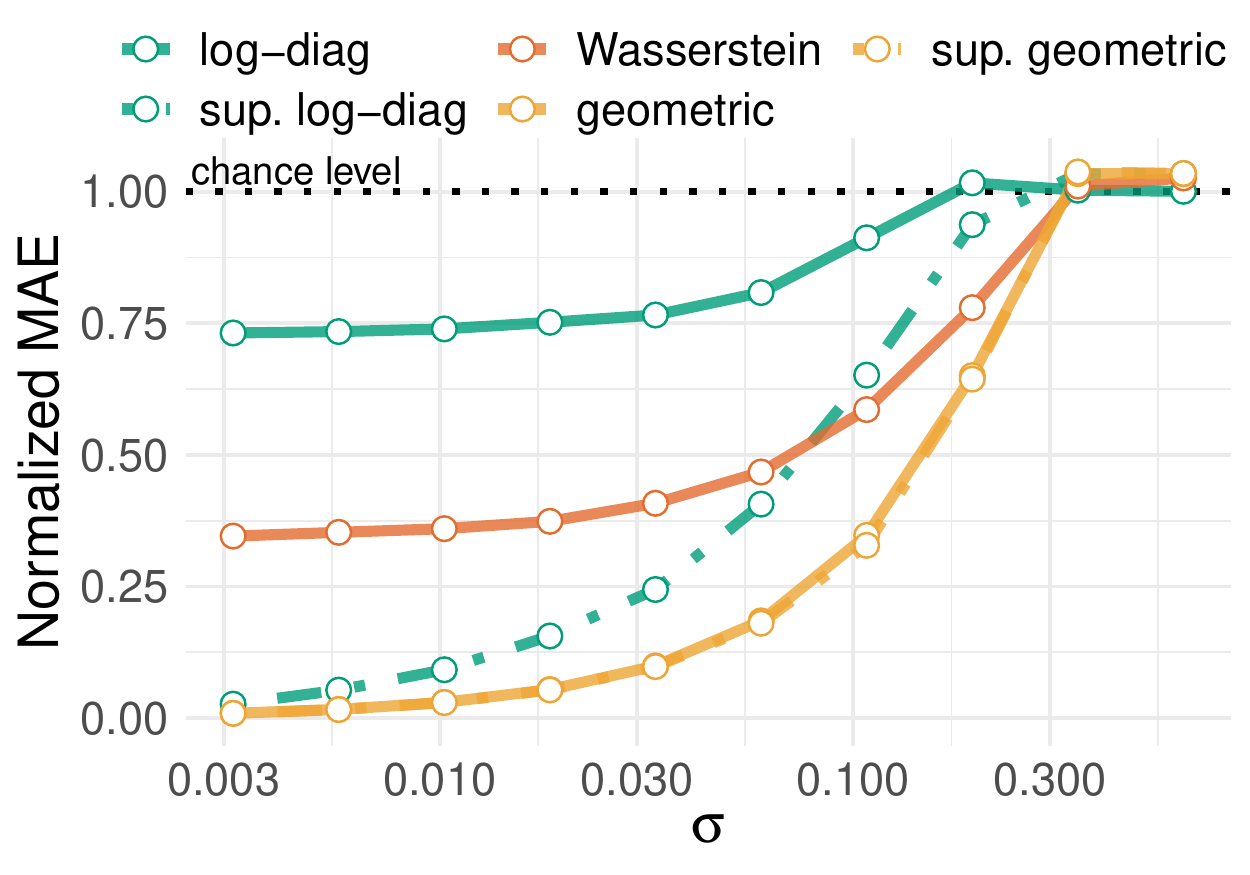}
        \end{minipage}
    \caption{Illustration of Prop.\ref{prop:consistency}.     
    Data is generated following the generative model with $f = \log$.  The
	    regression pipeline consists in projecting the data in the tangent
	    space, and then use a linear model.
	    The left plot shows the evolution of the score when random noise of
	    variance $\sigma^2$ is added to the variables $y_i$.  The MAE of
	    the geometric distance pipeline goes to $0$ in the limit of no
	    noise, indicating perfect out-of-sample prediction.  This
	    illustrates the linearity in the tangent space for the geometric
	    distance (Prop.~\ref{prop:consistency}).
	    The middle plot explores the effect of the parameter $\mu$
	    controlling the distance between $A$ and $I_P$.  Riemannian
	    geometric method is not affected by $\mu$ due to its affine
	    invariance property. Although the Wasserstein distance is not
	    affine invariant, its performance does not change much with $\mu$.
	    On the contrary, the log-diag method is sensitive to changes in
	    $A$.
	    The right plot shows how the score changes when mixing matrices
	    become sample dependent. We can see then only when $\sigma=0$
	    supervised + log-diag and Riemann reach perfect performance.
	    Geometric Riemann is uniformly better and indifferent to
	    projection choice. Wasserstein, despite model mismatch, outperforms
	    supervised + log-diag with high $\sigma$.}
    \label{fig:synth_expe}
    \end{minipage}
\end{figure}

\subsection{MEG data }
\paragraph{Predicting biological age from MEG on the Cambridge center of ageing dataset} In the following, we apply our methods to infer age from brain signals. Age is a dominant driver of cross-person variance in neuroscience data and a serious confounder~\cite{smith2018statistical}. As a consequence of the globally increased average lifespan, ageing has become a central topic in public health that has stimulated neuropsychiatric research at large scales. The link between age and brain function is therefore of utmost practical interest in neuroscientific research.

To predict age from brain signals, here we use the currently largest publicly available MEG dataset provided by the Cam-CAN~\cite{shafto2014cambridge}. We only considered the signals from magnetometer sensors ($P=102$) as it turns out that once SSS is applied (detailed in
Appendix \ref{subsec:appendix_preproc}), magnetometers and gradiometers are linear combination of approximately 70 signals ($65 \leq R_{i} \leq 73$), which become redundant
in practice~\cite{garces2017choice}.
%
We considered task-free recordings during which participants were asked to sit
still with eyes closed in the absence of systematic stimulation. We then drew $T \simeq 520,000$ time samples from $N=595$ subjects. 
To capture age-related changes in cortical brain rhythms~\cite{berthouze2010human,voytek2015age, clark2004spontaneous}, we filtered the data into $9$ frequency bands: low frequencies $[0.1 - 1.5]$,
$\delta[1.5 - 4]$, $\theta[4 - 8]$, $\alpha[8 - 15]$, $\beta_{low}[15 - 26]$,
$\beta_{high}[26 - 35]$, $\gamma_{low}[35 - 50]$, $\gamma_{mid}[50 - 74]$ and $\gamma_{high}[76 - 120]$ (Hz unit). These frequencies are compatible with conventional definitions used in the Human Connectome Project
\cite{larson2013adding}.
We verify that the covariance matrices all lie on a small portion of the manifold, justifying projection in a common tangent space.
Then we applied the covariance pipeline independently in each frequency band and concatenated the ensuing features.

\paragraph{Data-driven covariance projection for age prediction}

Three types of approaches are here compared: Riemannian methods (Wasserstein or geometric), methods extracting log-diagonal of matrices (with or without supervised spatial filtering, see Sec.~\ref{sec:spatialfilter}) and a biophysics-informed method based on the MNE source imaging technique~\cite{Hamalainen:1984}. The MNE method essentially consists in a standard Tikhonov regularized inverse solution and is therefore linear (See Appendix~\ref{subsec:appendix_mne} for details). Here it serves as gold-standard informed by the individual anatomy of each subject. It requires a T1-weighted MRI and the precise measure of the head in the MEG device coordinate system~\cite{baillet2017} and the coordinate alignment is hard to automate.
We configured MNE with $Q=8196$ candidate dipoles.
To obtain spatial smoothing and reduce dimensionality, we averaged the MNE solution using a cortical parcellation encompassing 448 regions of interest from~\cite{khan2018maturation,mne}. 
%
%
We then used ridge regression and tuned its regularization parameter by generalized cross-validation~\cite{Golub_Wahba:79} on a logarithmic grid of $100$ values in $[10^{-5}, 10^3]$ on each training fold of a 10-fold cross-validation loop.
%
%
All numerical experiments were run using the Scikit-Learn
software~\cite{sklearn}, the MNE software for processing M/EEG data~\cite{mne}
and the PyRiemann package~\cite{Congedo:2013}.
We also ported to Python some part of the Matlab code of Manopt toolbox~\cite{manopt} for computations involving Wasserstein distance.
The proposed method, including all data preprocessing, applied on the 500GB of raw MEG data from the Cam-CAN dataset, runs in approximately 12~hours on a regular desktop computer with at least 16GB of RAM. The preprocessing for the computation of the covariances is embarrassingly parallel and can therefore be significantly accelerated by using multiple CPUs. The actual predictive modeling can be performed in less than a minute on standard laptop.
Code used for data analysis can be found on GitHub\footnote{
\href{https://www.github.com/DavidSabbagh/NeurIPS19\_manifold-regression-meeg}{https://www.github.com/DavidSabbagh/NeurIPS19\_manifold-regression-meeg}}.

\paragraph{Riemannian projections are the leading data-driven methods}

\begin{figure}[h!]
    \begin{minipage}{\linewidth}
        \begin{minipage}{0.58\linewidth}
        \includegraphics[width=\linewidth]{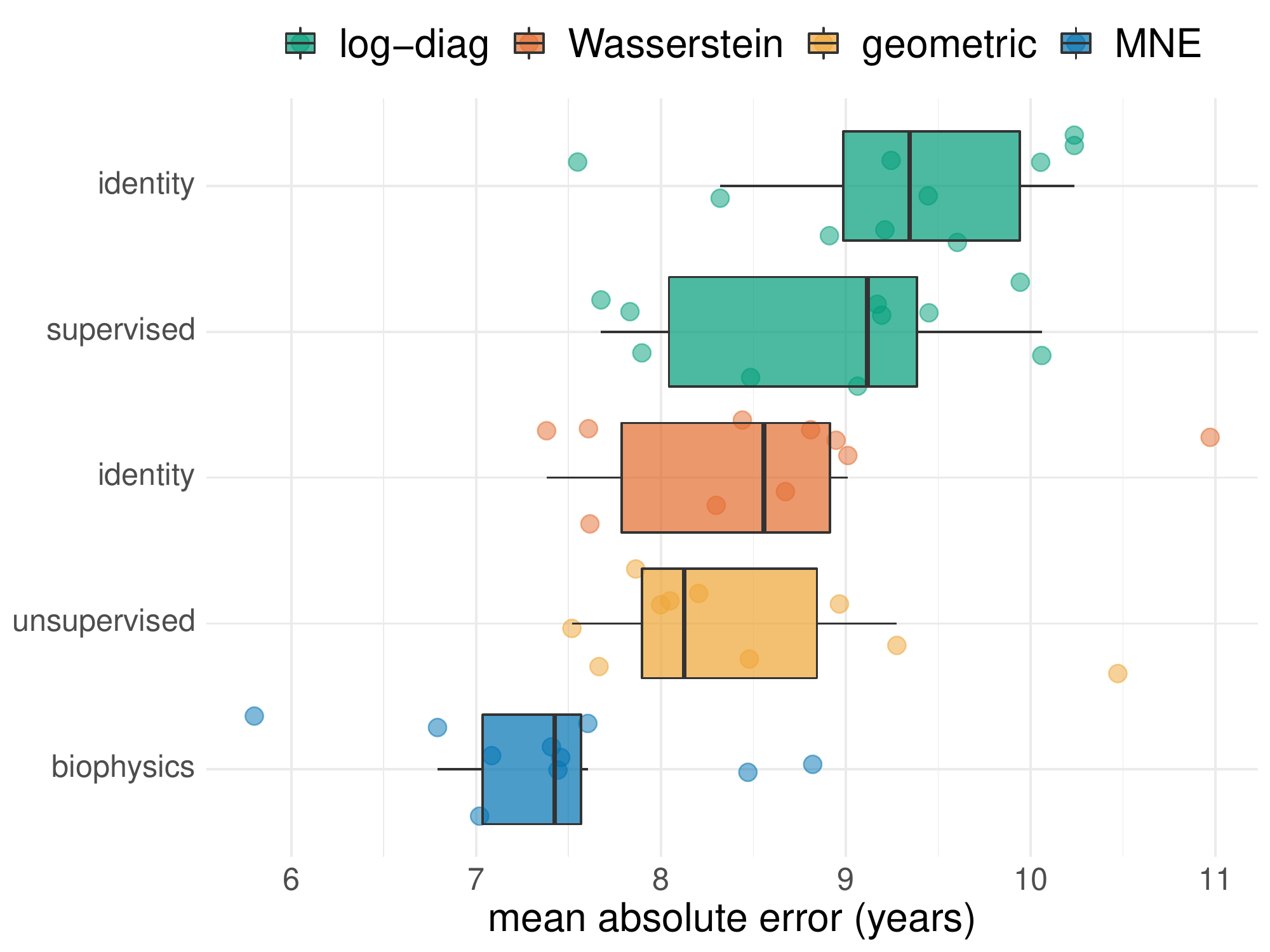}
        \end{minipage}
        \begin{minipage}{0.4\linewidth}
        \caption{Age prediction on Cam-CAN MEG dataset for different methods, ordered by out-of-sample MAE. The y-axis depicts the projection method, with identity denoting the absence of projection. Colors indicate the subsequent embedding. The biophysics-driven MNE method (blue) performs best. The Riemannian methods (orange) follow closely and their performance depends little on the projection method. The non-Riemannian methods log-diag (green) perform worse, although the supervised projection clearly helps.}
        \label{fig:real_expe}
        \end{minipage}
    \end{minipage}
\end{figure}

Fig.~\ref{fig:real_expe} displays the scores for each method. 
The biophysically motivated MNE projection yielded the best performance (7.4y MAE), closely followed by the purely data-driven Riemannian methods (8.1y MAE). The chance level was 16y MAE.
%
Interestingly, the Riemannian methods give similar results, and outperformed the non-Riemannian methods.
When Riemannian geometry was not applied, the projection strategy turned out to
be decisive.
Here, the supervised method performed best: it reduced the dimension of the problem while preserving the age-related variance.

Rejecting a null-hypothesis that differences between models are due to chance would require several
independent datasets. Instead, for statistical inference, we considered uncertainty estimates of paired
differences using 100 Monte Carlo splits (10\% test set size).
For each method, we counted how often it was performing better than
the baseline model obtained with identity and log-diag.
We observed for supervised log-diag 73\%,
identity Wasserstein 85\%, unsupervised geometric 96\% and biophysics 95\% improvement over baseline. This suggests that
inferences will carry over to new data.

Importantly, the supervised spatial filters and MNE both support model inspection, which is not the case for the two Riemannian methods. Fig. \ref{fig:meg_patterns} depicts the marginal patterns~\cite{haufe2014interpretation} from the supervised filters and the source-level ridge model, respectively. The sensor-level results suggest predictive dipolar patterns in the theta to beta range roughly compatible with generators in visual, auditory and motor cortices. Note that differences in head-position can make the sources appear deeper than they are (distance between the red positive and the blue negative poles). Similarly, the MNE-based model suggests localized predictive differences between frequency bands highlighting auditory, visual and premotor cortices. While the MNE model supports more exhaustive inspection, the supervised patterns are still physiologically informative. For example, one can notice that the pattern is more anterior in the $\beta$-band than the $\alpha$-band, potentially revealing sources in the motor cortex.

\begin{figure}[h!]
    \begin{minipage}{\linewidth}
        \begin{minipage}{0.68\linewidth}
        \includegraphics[width=\linewidth]{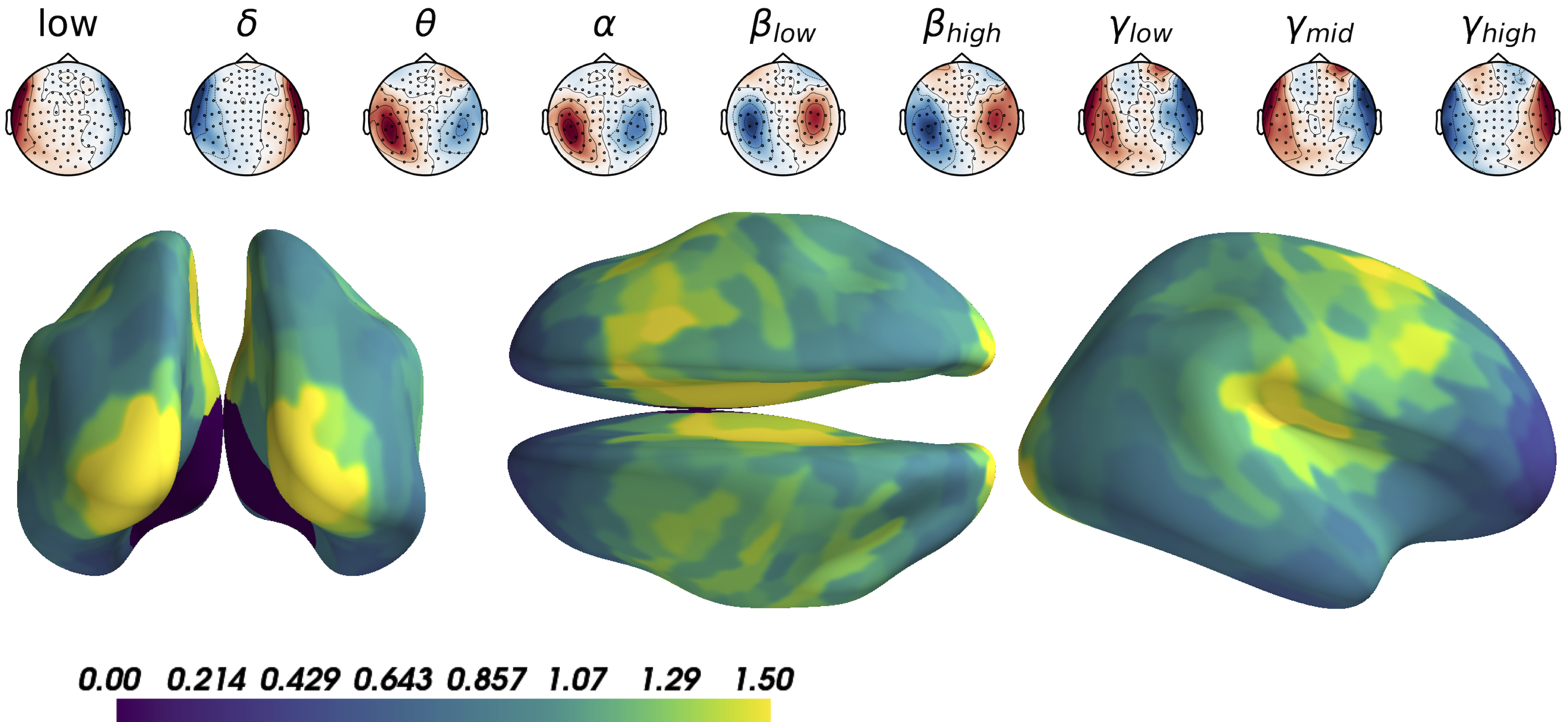}
        \end{minipage}
        \hspace{0.3em}
        \begin{minipage}{0.3\linewidth}
        \caption{Model inspection. Upper panel: sensor-level patterns from supervised projection. One can notice dipolar configurations varying across frequencies. Lower panel: standard deviation of patterns over frequencies from MNE projection highlighting bilateral visual, auditory and premotor cortices.
        } \label{fig:meg_patterns}
        \end{minipage}
    \end{minipage}
\end{figure}

\section{Discussion}

In this contribution, we proposed a mathematically principled approach for
regression on rank-reduced covariance matrices from M/EEG data. We applied this
framework to the problem of inferring age from neuroimaging data, for which we
made use of the currently largest publicly available MEG dataset. 
To the best of our knowledge, this is the first study to apply
a covariance-based approach coupled with Riemannian geometry
to regression problem in which the target is defined across persons and not
within persons (as in brain-computer interfaces).
%
%
Moreover, this study reports the first benchmark of age prediction from MEG
resting state data on the Cam-CAN. Our results demonstrate that Riemannian
data-driven methods do not fall far behind the gold-standard methods with
biophysical priors, that depend on manual data processing. 
One limitation of Riemannian methods is, however, their interpretability compared to
other models that allow to report brain-region and frequency-specific effects. These results suggest a trade-off between performance and explainability. Our study suggests that the Riemannian methods have the potential to support automated large-scale analysis of M/EEG data in the absence of MRI scans. Taken together, this potentially opens new avenues for biomarker development. 




\section*{Acknowledgement}{}
This work was supported by a 2018 \emph{``m\'edecine num\'erique'' (for digital Medicine)} thesis grant issued by Inserm (French national institute of health and medical research) and Inria (French national research institute for the digital sciences). It was also partly supported by the European
Research Council Starting Grant SLAB ERC-YStG-676943.

\bibliography{references_camera_ready}

\begin{thebibliography}{10}

\bibitem{absil2009optimization}
P-A Absil, Robert Mahony, and Rodolphe Sepulchre.
\newblock {\em Optimization algorithms on matrix manifolds}.
\newblock Princeton University Press, 2009.

\bibitem{lemon}
Anahit Babayan, Miray Erbey, Deniz Kumral, Janis~D. Reinelt, Andrea M.~F.
  Reiter, Josefin R{\"o}bbig, H.~Lina Schaare, Marie Uhlig, Alfred Anwander,
  Pierre-Louis Bazin, Annette Horstmann, Leonie Lampe, Vadim~V. Nikulin, Hadas
  Okon-Singer, Sven Preusser, Andr{\'e} Pampel, Christiane~S. Rohr, Julia
  Sacher, Angelika Th{\"o}ne-Otto, Sabrina Trapp, Till Nierhaus, Denise
  Altmann, Katrin Arelin, Maria Bl{\"o}chl, Edith Bongartz, Patric Breig, Elena
  Cesnaite, Sufang Chen, Roberto Cozatl, Saskia Czerwonatis, Gabriele
  Dambrauskaite, Maria Dreyer, Jessica Enders, Melina Engelhardt, Marie~Michele
  Fischer, Norman Forschack, Johannes Golchert, Laura Golz, C.~Alexandrina
  Guran, Susanna Hedrich, Nicole Hentschel, Daria~I. Hoffmann, Julia~M.
  Huntenburg, Rebecca Jost, Anna Kosatschek, Stella Kunzendorf, Hannah Lammers,
  Mark~E. Lauckner, Keyvan Mahjoory, Ahmad~S. Kanaan, Natacha Mendes, Ramona
  Menger, Enzo Morino, Karina N{\"a}the, Jennifer Neubauer, Handan Noyan,
  Sabine Oligschl{\"a}ger, Patricia Panczyszyn-Trzewik, Dorothee Poehlchen,
  Nadine Putzke, Sabrina Roski, Marie-Catherine Schaller, Anja Schieferbein,
  Benito Schlaak, Robert Schmidt, Krzysztof~J. Gorgolewski, Hanna~Maria
  Schmidt, Anne Schrimpf, Sylvia Stasch, Maria Voss, Annett Wiedemann,
  Daniel~S. Margulies, Michael Gaebler, and Arno Villringer.
\newblock {A mind-brain-body dataset of MRI, EEG, cognition, emotion, and
  peripheral physiology in young and old adults}.
\newblock {\em Scientific Data}, 6:180308 EP --, 02 2019.

\bibitem{baillet2017}
Sylvain Baillet.
\newblock Magnetoencephalography for brain electrophysiology and imaging.
\newblock {\em Nature Neuroscience}, 20:327 EP --, 02 2017.

\bibitem{berthouze2010human}
Luc Berthouze, Leon~M James, and Simon~F Farmer.
\newblock Human eeg shows long-range temporal correlations of oscillation
  amplitude in theta, alpha and beta bands across a wide age range.
\newblock {\em Clinical Neurophysiology}, 121(8):1187--1197, 2010.

\bibitem{Bhatia:07}
Rajendra Bhatia.
\newblock {\em Positive Definite Matrices}.
\newblock Princeton University Press, 2007.

\bibitem{bhatia2018bures}
Rajendra Bhatia, Tanvi Jain, and Yongdo Lim.
\newblock On the bures--wasserstein distance between positive definite
  matrices.
\newblock {\em Expositiones Mathematicae}, 2018.

\bibitem{blankertz-etal:08}
B.~{Blankertz}, R.~{Tomioka}, S.~{Lemm}, M.~{Kawanabe}, and K.~{Muller}.
\newblock Optimizing spatial filters for robust eeg single-trial analysis.
\newblock {\em IEEE Signal Processing Magazine}, 25(1):41--56, 2008.

\bibitem{bonnabel2009riemannian}
Silvere Bonnabel and Rodolphe Sepulchre.
\newblock Riemannian metric and geometric mean for positive semidefinite
  matrices of fixed rank.
\newblock {\em SIAM Journal on Matrix Analysis and Applications},
  31(3):1055--1070, 2009.

\bibitem{manopt}
N.~Boumal, B.~Mishra, P.-A. Absil, and R.~Sepulchre.
\newblock {M}anopt, a {M}atlab toolbox for optimization on manifolds.
\newblock {\em Journal of Machine Learning Research}, 15:1455--1459, 2014.

\bibitem{Buzsaki:2017}
Gy{\"o}rgy Buzs{\'a}ki and Rodolfo Llin{\'a}s.
\newblock Space and time in the brain.
\newblock {\em Science}, 358(6362):482--485, 2017.

\bibitem{buzsaki2014log}
Gy{\"o}rgy Buzs{\'a}ki and Kenji Mizuseki.
\newblock The log-dynamic brain: how skewed distributions affect network
  operations.
\newblock {\em Nature Reviews Neuroscience}, 15(4):264, 2014.

\bibitem{clark2004spontaneous}
C~Richard Clark, Melinda~D Veltmeyer, Rebecca~J Hamilton, Elena Simms, Robert
  Paul, Daniel Hermens, and Evian Gordon.
\newblock Spontaneous alpha peak frequency predicts working memory performance
  across the age span.
\newblock {\em International Journal of Psychophysiology}, 53(1):1--9, 2004.

\bibitem{Congedo:2013}
M.~{Congedo}, A.~{Barachant}, and A.~{Andreev}.
\newblock A new generation of brain-computer interface based on {R}iemannian
  geometry.
\newblock {\em arXiv e-prints}, October 2013.

\bibitem{congedo-etal:17}
Marco Congedo, Alexandre Barachant, and Rajendra Bhatia.
\newblock Riemannian geometry for {EEG}-based brain-computer interfaces; a
  primer and a review.
\newblock {\em Brain-Computer Interfaces}, 4(3):155--174, 2017.

\bibitem{dahne2014spoc}
Sven D{\"a}hne, Frank~C Meinecke, Stefan Haufe, Johannes H{\"o}hne, Michael
  Tangermann, Klaus-Robert M{\"u}ller, and Vadim~V Nikulin.
\newblock Spoc: a novel framework for relating the amplitude of neuronal
  oscillations to behaviorally relevant parameters.
\newblock {\em NeuroImage}, 86:111--122, 2014.

\bibitem{dmochowski2012correlated}
Jacek Dmochowski, Paul Sajda, Joao Dias, and Lucas Parra.
\newblock Correlated components of ongoing eeg point to emotionally laden
  attention – a possible marker of engagement?
\newblock {\em Frontiers in Human Neuroscience}, 6:112, 2012.

\bibitem{engemann2015automated}
Denis~A Engemann and Alexandre Gramfort.
\newblock Automated model selection in covariance estimation and spatial
  whitening of meg and eeg signals.
\newblock {\em NeuroImage}, 108:328--342, 2015.

\bibitem{forstner2003metric}
Wolfgang F{\"o}rstner and Boudewijn Moonen.
\newblock A metric for covariance matrices.
\newblock In {\em Geodesy-The Challenge of the 3rd Millennium}, pages 299--309.
  Springer, 2003.

\bibitem{garces2017choice}
Pilar Garc{\'e}s, David L{\'o}pez-Sanz, Fernando Maest{\'u}, and Ernesto
  Pereda.
\newblock Choice of magnetometers and gradiometers after signal space
  separation.
\newblock {\em Sensors}, 17(12):2926, 2017.

\bibitem{Golub_Wahba:79}
Gene~H. Golub, Michael Heath, and Grace Wahba.
\newblock Generalized cross-validation as a method for choosing a good ridge
  parameter.
\newblock {\em Technometrics}, 21(2):215--223, 1979.

\bibitem{mne}
Alexandre Gramfort, Martin Luessi, Eric Larson, Denis~A. Engemann, Daniel
  Strohmeier, Christian Brodbeck, Lauri Parkkonen, and Matti~S.
  H{\"a}m{\"a}l{\"a}inen.
\newblock {MNE} software for processing {MEG} and {EEG} data.
\newblock {\em NeuroImage}, 86:446--460, Feb. 2014.

\bibitem{grosse-wentrup:08}
M.~{Grosse-Wentrup*} and M.~{Buss}.
\newblock Multiclass common spatial patterns and information theoretic feature
  extraction.
\newblock {\em IEEE Transactions on Biomedical Engineering}, 55(8):1991--2000,
  Aug 2008.

\bibitem{hamalainen1993magnetoencephalography}
Matti H{\"a}m{\"a}l{\"a}inen, Riitta Hari, Risto~J Ilmoniemi, Jukka Knuutila,
  and Olli~V Lounasmaa.
\newblock Magnetoencephalography—theory, instrumentation, and applications to
  noninvasive studies of the working human brain.
\newblock {\em Reviews of modern Physics}, 65(2):413, 1993.

\bibitem{Hamalainen:1984}
MS~H{\"a}m{\"a}l{\"a}inen and RJ~Ilmoniemi.
\newblock Interpreting magnetic fields of the brain: minimum norm estimates.
\newblock Technical Report TKK-F-A559, Helsinki University of Technology, 1984.

\bibitem{harandi2017dimensionality}
Mehrtash Harandi, Mathieu Salzmann, and Richard Hartley.
\newblock Dimensionality reduction on spd manifolds: The emergence of
  geometry-aware methods.
\newblock {\em IEEE transactions on pattern analysis and machine intelligence},
  40(1):48--62, 2017.

\bibitem{hari2017meg}
Riitta Hari and Aina Puce.
\newblock {\em MEG-EEG Primer}.
\newblock Oxford University Press, 2017.

\bibitem{haufe2014interpretation}
Stefan Haufe, Frank Meinecke, Kai Görgen, Sven Dähne, John-Dylan Haynes,
  Benjamin Blankertz, and Felix Bießmann.
\newblock On the interpretation of weight vectors of linear models in
  multivariate neuroimaging.
\newblock {\em NeuroImage}, 87:96 -- 110, 2014.

\bibitem{horev2016geometry}
Inbal Horev, Florian Yger, and Masashi Sugiyama.
\newblock Geometry-aware principal component analysis for symmetric positive
  definite matrices.
\newblock {\em Machine Learning}, 106, 11 2016.

\bibitem{jas2017autoreject}
Mainak Jas, Denis~A Engemann, Yousra Bekhti, Federico Raimondo, and Alexandre
  Gramfort.
\newblock Autoreject: Automated artifact rejection for {MEG} and {EEG} data.
\newblock {\em NeuroImage}, 159:417--429, 2017.

\bibitem{journee2010low}
Michel Journ{\'e}e, Francis Bach, P-A Absil, and Rodolphe Sepulchre.
\newblock Low-rank optimization on the cone of positive semidefinite matrices.
\newblock {\em SIAM Journal on Optimization}, 20(5):2327--2351, 2010.

\bibitem{khan2018maturation}
Sheraz Khan, Javeria~A Hashmi, Fahimeh Mamashli, Konstantinos Michmizos,
  Manfred~G Kitzbichler, Hari Bharadwaj, Yousra Bekhti, Santosh Ganesan,
  Keri-Lee~A Garel, Susan Whitfield-Gabrieli, et~al.
\newblock Maturation trajectories of cortical resting-state networks depend on
  the mediating frequency band.
\newblock {\em NeuroImage}, 174:57--68, 2018.

\bibitem{larson2013adding}
Linda~J Larson-Prior, Robert Oostenveld, Stefania Della~Penna, G~Michalareas,
  F~Prior, Abbas Babajani-Feremi, J-M Schoffelen, Laura Marzetti, Francesco
  de~Pasquale, F~Di~Pompeo, et~al.
\newblock {Adding dynamics to the Human Connectome Project with MEG}.
\newblock {\em Neuroimage}, 80:190--201, 2013.

\bibitem{liem-etal:2017}
Franziskus Liem, Gaël Varoquaux, Jana Kynast, Frauke Beyer, Shahrzad~Kharabian
  Masouleh, Julia~M. Huntenburg, Leonie Lampe, Mehdi Rahim, Alexandre Abraham,
  R.~Cameron Craddock, Steffi Riedel-Heller, Tobias Luck, Markus Loeffler,
  Matthias~L. Schroeter, Anja~Veronica Witte, Arno Villringer, and Daniel~S.
  Margulies.
\newblock Predicting brain-age from multimodal imaging data captures cognitive
  impairment.
\newblock {\em NeuroImage}, 148:179 -- 188, 2017.

\bibitem{makeig-etal:1995}
Scott Makeig, Anthony~J. Bell, Tzyy-Ping Jung, and Terrence~J. Sejnowski.
\newblock Independent component analysis of electroencephalographic data.
\newblock In {\em Proceedings of the 8th International Conference on Neural
  Information Processing Systems}, NIPS'95, pages 145--151, Cambridge, MA, USA,
  1995. MIT Press.

\bibitem{massart2018psd}
Estelle Massart and Pierre-Antoine Absil.
\newblock Quotient geometry with simple geodesics for the manifold of
  fixed-rank positive-semidefinite matrices.
\newblock Technical report, UCLouvain, 2018.
\newblock preprint on webpage at \url{http://sites.uclouvain.be/absil/2018.06}.

\bibitem{sklearn}
F.~Pedregosa, G.~Varoquaux, A.~Gramfort, V.~Michel, B.~Thirion, O.~Grisel,
  M.~Blondel, P.~Prettenhofer, R.~Weiss, V.~Dubourg, J.~Vanderplas, A.~Passos,
  D.~Cournapeau, M.~Brucher, M.~Perrot, and E.~Duchesnay.
\newblock Scikit-learn: Machine learning in {P}ython.
\newblock {\em Journal of Machine Learning Research}, 12:2825--2830, 2011.

\bibitem{rodrigues2018multivariate}
Pedro Luiz~Coelho Rodrigues, Marco Congedo, and Christian Jutten.
\newblock Multivariate time-series analysis via manifold learning.
\newblock In {\em 2018 IEEE Statistical Signal Processing Workshop (SSP)},
  pages 573--577. IEEE, 2018.

\bibitem{shafto2014cambridge}
Meredith~A Shafto, Lorraine~K Tyler, Marie Dixon, Jason~R Taylor, James~B Rowe,
  Rhodri Cusack, Andrew~J Calder, William~D Marslen-Wilson, John Duncan, Tim
  Dalgleish, et~al.
\newblock {The Cambridge Centre for Ageing and Neuroscience (Cam-CAN) study
  protocol: a cross-sectional, lifespan, multidisciplinary examination of
  healthy cognitive ageing}.
\newblock {\em BMC neurology}, 14(1):204, 2014.

\bibitem{smith2018statistical}
Stephen~M Smith and Thomas~E Nichols.
\newblock Statistical challenges in “big data” human neuroimaging.
\newblock {\em Neuron}, 97(2):263--268, 2018.

\bibitem{taulu}
Samu Taulu and Matti Kajola.
\newblock Presentation of electromagnetic multichannel data: the signal space
  separation method.
\newblock {\em Journal of Applied Physics}, 97(12):124905, 2005.

\bibitem{taylor2017cambridge}
Jason~R Taylor, Nitin Williams, Rhodri Cusack, Tibor Auer, Meredith~A Shafto,
  Marie Dixon, Lorraine~K Tyler, Richard~N Henson, et~al.
\newblock {The Cambridge Centre for Ageing and Neuroscience (Cam-CAN) data
  repository: structural and functional MRI, MEG, and cognitive data from a
  cross-sectional adult lifespan sample}.
\newblock {\em Neuroimage}, 144:262--269, 2017.

\bibitem{uusitalo1997signal}
Mikko~A Uusitalo and Risto~J Ilmoniemi.
\newblock Signal-space projection method for separating {MEG} or {EEG} into
  components.
\newblock {\em Medical and Biological Engineering and Computing},
  35(2):135--140, 1997.

\bibitem{vandereycken2009embedded}
Bart Vandereycken, P-A Absil, and Stefan Vandewalle.
\newblock Embedded geometry of the set of symmetric positive semidefinite
  matrices of fixed rank.
\newblock In {\em 2009 IEEE/SP 15th Workshop on Statistical Signal Processing},
  pages 389--392. IEEE, 2009.

\bibitem{voytek2015age}
Bradley Voytek, Mark~A Kramer, John Case, Kyle~Q Lepage, Zechari~R Tempesta,
  Robert~T Knight, and Adam Gazzaley.
\newblock Age-related changes in 1/f neural electrophysiological noise.
\newblock {\em Journal of Neuroscience}, 35(38):13257--13265, 2015.

\bibitem{yger-etal:17}
F.~{Yger}, M.~{Berar}, and F.~{Lotte}.
\newblock Riemannian approaches in brain-computer interfaces: A review.
\newblock {\em IEEE Transactions on Neural Systems and Rehabilitation
  Engineering}, 25(10):1753--1762, Oct 2017.

\end{thebibliography}

\newpage
\section{Appendix} \label{sec:appendix}

\subsection{Proof of proposition~\ref{prop:consistency}} \label{subsec:proof_prop}

First, we note that by invariance, $\overline{\bm{C}} = \text{Mean}_G(\bm{C}_1, \dots, \bm{C}_N) = \bm{A} \text{Mean}_G(\bm{E}_1, \dots, \bm{E}_N)\bm{A}^{\top} = \bm{A}\overline{\bm{E}}\bm{A}^{\top}$, where $\overline{\bm{E}}$ has the same block diagonal structure as the $\bm{E}_i$'s, and $\overline{\bm{E}}_{jj} =  (\prod_{i=1}^Np_{i, j})^{\frac{1}{N}}$ for $j \leq Q$.
Denote $\overline{\bm{U}} = \overline{\bm{C}}^{\frac12}\bm{A}^{-\top}\overline{\bm{E}}^{-\frac12}$.
By simple verification, we obtain $\overline{\bm{U}}^{\top}\overline{\bm{U}} = I_P$, i.e. $\overline{\bm{U}}$ is orthogonal.

Furthermore, we have:
$$
\overline{\bm{U}}^{\top}\overline{\bm{C}}^{-\frac12}\bm{C}_i\overline{\bm{C}}^{-\frac12}\overline{\bm{U}} = \overline{\bm{E}}^{-\frac12}\bm{E}_i\overline{\bm{E}}^{-\frac12} \enspace .
$$
It follows that for all $i$, 
$$
\overline{\bm{U}}^{\top}\log(\overline{\bm{C}}^{-\frac12}\bm{C}_i\overline{\bm{C}}^{-\frac12})\overline{\bm{U}} = \log(\overline{\bm{E}}^{-\frac12}\bm{E}_i\overline{\bm{E}}^{-\frac12})
$$
Note that $\log(\overline{\bm{E}}^{-\frac12}\bm{E}_i\overline{\bm{E}}^{-\frac12})$ shares the same structure as the $\bm{E}_i$'s, and that $\log(\overline{\bm{E}}^{-\frac12}\bm{E}_i\overline{\bm{E}}^{-\frac12})_{jj} = \log(\frac{p_{i,j}}{\bar{p}_{j}})$. for $j \leq Q$.

Therefore, the relationship between $ \log(\overline{\bm{C}}^{-\frac12}\bm{C}_i\overline{\bm{C}}^{-\frac12})$ and the $\log(p_{i,j})$ is linear.

Finally, since $\bm{v}_i = \text{Upper}(\log(\overline{\bm{C}}^{-\frac12}\bm{C}_i\overline{\bm{C}}^{-\frac12}))$, the relationship between the $\bm{v}_i$'s and the $\log(p_{i,j})$ is linear, and the result holds.

\subsection{Proof of proposition~\ref{prop:consistency_wass}} \label{subsec:proof_prop_wass}

First, we note that $\bm{C}_i = \bm{A}\bm{E}_i\bm{A}^{\top} \in \SpPP = \mathcal{S}_{P,P}^{+}$ so it can be decomposed as $\bm{C}_i = \bm{Y}_i \bm{Y}_i^{\top}$ with $\bm{Y}_i =\bm{A}\bm{E}_i^{\frac12}$. 

By orthogonal invariance, $\overline{\bm{C}} = \text{Mean}_W(\bm{C}_1, \dots, \bm{C}_N) = \bm{A} \text{Mean}_W(\bm{E}_1, \dots, \bm{E}_N)\bm{A}^{\top} = \bm{A}\overline{\bm{E}}\bm{A}^{\top}$, where $\overline{\bm{E}}$ so has the same block diagonal structure as the $\bm{E}_i$'s, and $\overline{\bm{E}}_{jj} =  (\sum_i \sqrt{p_{ij}})^2$ for $j \leq Q$. $\overline{\bm{C}}$ is also decomposed as $\overline{\bm{C}} = \overline{\bm{Y}} \overline{\bm{Y}}^{\top}$ with $\overline{\bm{Y}}=\bm{A}\overline{\bm{E}}^{\frac12}$.

Further, $\bm{Q}_i^* = \bm{V}_i\bm{U}_i^{\top}$ with $\bm{U}_i$ and $\bm{V}_i$ coming from the SVD of $\overline{\bm{Y}}^{\top} \bm{Y}_i = \overline{\bm{E}}^{\frac12} \bm{E}_i^{\frac12}$ which has the same structure as the $\bm{E}_i$'s. Therefore $\bm{Q}_i^*$ has also the same structure with the identity matrix as its upper block. 

Finally we have $\bm{v}_i = \mathcal{P}_{\overline{\bm{C}}}(\bm{C}_i) = 
\vect(\bm{Y}_i \bm{Q}_i^* - \overline{\bm{Y}})$ so it is linear in $\sqrt{(p_{i,j})}$ for $j \leq Q$.

\subsection{Proof that there is no continuous affine invariant distance on $S_{P, R}^+$ if $R < P$} \label{subsec:noaff_inv_distance}

We show the result for $P=2$ and $R=1$; the demonstration can straightforwardly be extended to the other cases.
The proof, from~\cite{bonnabel2009riemannian}, is by contradiction.

Assume that $d$ is a continuous invariant distance on $S_{2, 1}^+ $.
Consider $\bm{A} = \begin{pmatrix} 
1 & 0 \\
0 & 0 
\end{pmatrix}$ and $\bm{B} = \begin{pmatrix} 
1 & 1 \\
1 & 1 
\end{pmatrix}$, both in $S_{2, 1}^+$.
For $\varepsilon > 0$, consider the invertible matrix $\bm{W}_{\varepsilon} = \begin{pmatrix} 
1 & 0 \\
0 & \varepsilon 
\end{pmatrix}$.

We have: $\bm{W}_{\varepsilon}\bm{A}\bm{W}_{\varepsilon}^{\top} = \bm{A}$, and $\bm{W}_{\varepsilon}\bm{B}\bm{W}_{\varepsilon}^{\top} = \begin{pmatrix} 
1 & \varepsilon \\
\varepsilon & \varepsilon^2 
\end{pmatrix}$.

Hence, as $\varepsilon$ goes to $0$, we have $\bm{W}_{\varepsilon}\bm{B}\bm{W}_{\varepsilon}^{\top} \to \bm{A}$

Using affine invariance, we have: 
$$ d(\bm{A}, \bm{B}) = d(\bm{W}_{\varepsilon}\bm{A}\bm{W}_{\varepsilon}^{\top}, \bm{W}_{\varepsilon}\bm{B}\bm{W}_{\varepsilon}^{\top}) \enspace$$

Letting $\varepsilon \to 0$  and using continuity of $d$ yields $d(\bm{A}, \bm{B}) = d(\bm{A}, \bm{A}) = 0$, which is absurd since $\bm{A} \neq \bm{B}$.

\subsection{Supervised Spatial Filtering} \label{subsec:appendix_spoc}
We assume that the signal $\bm{x}(t)$ is band-pass filtered in one of frequency band of interest, so that for each subject the band power of signal is approximated by the variance over time of the signal.
We denote the expectation $\bbE$ and the variance $\VAR$ over time $t$ or subject $i$ by a corresponding subscript.

The source extracted by a spatial filter $\bm{w}$ for subject $i$ is $\bm{\widehat{s}}_i =
\bm{w}^{\top}\bm{x}_i(t)$. Its power reads:
\begin{align*}
    \Phi_{i}^{\bm{w}} &= \VAR_{t}[\bm{w}^{\top}\bm{x}_{i}(t)] =
    \bbE_{t}[\bm{w}^{\top}\bm{x}_{i}(t) \bm{x}_{i}^{\top}(t) \bm{w}] = \bm{w}^{\top} \bm{C}_i \bm{w}
\end{align*}
and its expectation across subjects is given by:
\begin{align*}
\bbE_{i}[\Phi_{i}^{\bm{w}}] &= \bm{w}^{\top} \bbE_{i}[\bm{C}_i]
    \bm{w}  = \bm{w}^{\top} \overline{\bm{C}} \bm{w} \enspace,
\end{align*}
where $\overline{\bm{C}} = \frac{1}{N} \sum_{i} \bm{C}_i$ is the average covariance matrix across subjects. Note that here, $\bm{C}_i$ refers to the covariance of the $\bm{x}_{i}$ and not its estimate as in Sec.~\ref{sec:spatialfilter}.

We aim to maximize the covariance between the target $y$ and the power of the sources, $\COV_{i}[\Phi_{i}^{\bm{w}}, y_{i}]$. This quantity is affected by the scaling of its arguments. To address this, the target variable $y$ is normalized:
$$
    \bbE_{i}[y_{i}] = 0 \quad \VAR_{i}[y_{i}] = 1 \enspace .
$$
Following \citep{dahne2014spoc}, to also scale $\Phi_{i}^{\bm{w}}$ we constrain its expectation to be 1:
$$
    \bbE_{i}[\Phi_{i}^{\bm{w}}] = \bm{w}^{\top} \overline{\bm{C}} \bm{w} = 1
$$
The quantity one aims to maximize reads:
\begin{align*}
        \COV_{i}[\Phi_{i}^{\bm{w}}, y_{i}] &= \bbE_{i}[\
        (\Phi_{i}^{\bm{w}} - \bbE_{i}[\Phi_{i}^{\bm{w}}])\ (y_{i} -
        \bbE_{i}[y_{i}])\ ]\\
        &= \bm{w}^{\top} \bbE_{i}[\bm{C}_i y_{i}] \bm{w}  - \bm{w}^{\top} \overline{\bm{C}} \bm{w} \bbE_{i}[y_{i}]\\
        &= \bm{w}^{\top} \bm{C}_{y} \bm{w}
\end{align*}
where $\bm{C}_{y} = \frac{1}{N} \sum_{i} y_{i} \bm{C}_i$.\\
\\
Taking into account the normalization constraint we obtain:
\begin{align}
    \label{eq:objconst}
    \widehat{\bm{w}} = \argmax_{\bm{w}^{\top} \overline{\bm{C}} \bm{w} = 1} \bm{w}^{\top} \bm{C}_{y} \bm{w} \enspace .
\end{align}
The Lagrangian of \eqref{eq:objconst} reads $F(\bm{w}, \lambda) = \bm{w}^{\top}
\bm{C}_{y} \bm{w} + \lambda (1 - \bm{w}^{\top} \overline{\bm{C}} \bm{w})$.
Setting its gradient \wrt $\bm{w}$ to $0$ yields a generalized
eigenvalue problem:
\begin{align}
    \nabla_{\bm{w}} F(\bm{w}, \lambda) = 0 &
        \implies \Sigma_{y} \bm{w} = \lambda
            \overline{\Sigma_{\bm{x}}} \bm{w} \label{eq:eigen1}
\end{align}
Note that \eqref{eq:objconst} can be also written as a generalized Rayleigh quotient:
\begin{align*}
    \widehat{\bm{w}} = \argmax_{\bm{w}}
        \frac{
            \bm{w}^{\top} \bm{C}_{y} \bm{w}
        }{
            \bm{w}^{\top} \overline{\bm{C}} \bm{w}
        } \enspace .
\end{align*}
Equation \eqref{eq:eigen1} has a unique closed-form solution called the generalized
eigenvectors of $(\bm{C}_{y}, \overline{\bm{C}})$. The second derivative gives:
\begin{align}
    \nabla_{\bm{\lambda}} F(\bm{w}, \lambda) = 0 &
        \implies \lambda = \bm{w}^{\top}
            \Sigma_{y} \bm{w} = \COV_{i}[\Phi_{i}^{\bm{w}}, y_{i}] \label{eq:eigen2}
\end{align}
Equation \eqref{eq:eigen2} leads to an interpretation of $\lambda$ as the covariance between $\Phi^{\bm{w}}$ and $y$, which should be maximal.
As a consequence, $\bm{W}_{\text{SUP}}$ is built from the generalized eigenvectors of
Eq.\eqref{eq:eigen1}, sorted by decreasing eigenvalues.

\subsection{MNE-based spatial filtering} \label{subsec:appendix_mne}

Let us denote $\bm{G} \in \bbR^{P \times Q}$ the instantaneous mixing matrix that relates the sources in the brain to the MEG/EEG measurements. This forward operator matrix is obtained by solving numerically Maxwell's equations after specifying a geometrical model of the head, typically obtained using an anatomical MRI image~\cite{hari2017meg}. Here $Q \geq P$ corresponds to the number of candidate sources in the brain. The MNE approach~\cite{Hamalainen:1984} offers a way to solve the inverse problem. MNE can be seen as Tikhonov regularized estimation, also similar to a ridge regression in statistics. Using such problem formulation the sources are obtained from the measurements with a linear operator which is given by:
$$
    \bm{W}_{\text{MNE}} = \bm{G}^\top(\bm{G}\bm{G}^\top + \lambda \bm{I}_P)^{-1} \in \bbR^{Q \times P} \enspace .
$$
The rows of this linear operator $\bm{W}_{\text{MNE}}$ can be seen also as spatial filters
that are mapped to specific locations in the brain. These are the filters used in Fig.~\ref{fig:real_expe}, using the implementation from~\cite{mne}.

\subsection{Preprocessing} \label{subsec:appendix_preproc}
Typical brain's magnetic fields detected by MEG are in the order of 100
femtotesla ($1 fT = 10^{-15}$ T) which is \textasciitilde $10^{-8}$ times the
strength of the earth's steady magnetic field.  That is why MEG recordings are
carried out inside special magnetically shielded rooms (MSR) that eliminate or
at least dampen external ambient magnetic disturbances. 

To pick up such tiny magnetic fields sensitive sensors have to be used \cite{hari2017meg}.
Their extreme sensitivity is challenged by many
electromagnetic nuisance sources (any moving metal objects like cars or
elevators) or electrically powered instruments generating magnetic induction
that is orders of magnitude stronger than the brain's. Their influence can be
reduced by combining magnetometers coils (that directly record the magnetic
field) with gradiometers coils (that record the gradient of the magnetic field
in certain directions). Those gradiometers, arranged either in a radial or
tangential (planar) way, record the gradient of the magnetic field towards 2
perpendicular directions hence inherently greatly emphasize brain signals with
respect to environmental noise.

Even though the magnetic shielded room and gradiometer coils can help to reduce
the effects of external interference signals the problem mainly remains and
further reduction is needed.  Also additional artifact signals can be caused by
movement of the subject during recording if the subject has small magnetic
particles on his body or head.  The Signal Space Separation (SSS) method can help mitigate those problems~\cite{taulu}.

\paragraph{Signal Space Separation (SSS)}

The Signal Space Separation (SSS) method \cite{taulu}, also called Maxwell
Filtering, is a biophysical spatial filtering method that aim to produce
signals cleaned from external interference signals and from
movement distortions and artifacts.

A MEG device records the neuromagnetic field distribution by sampling the field
simultaneously at P distinct locations around the subject’s head.  At each
moment of time the measurement is a vector $\bm{x} \in \bbR^P$ is the total
number of recording channels.

In theory, any direction of this vector in the signal space represents a valid
measurement of a magnetic field, however the knowledge of the location of
possible sources of magnetic field, the geometry
of the sensor array and electromagnetic theory (Maxwell's equations and the
quasistatic approximation) considerably constrain the
relevant signal space and allow us to differentiate between signal space
directions consistent with a brain's field and those that are not.

To be more precise, it has been shown that the recorded magnetic field is a
gradient of a harmonic scalar potential.  A harmonic potential $V(\bm{r})$ is a
solution of the Laplacian differential equation $\nabla^{2} V = 0 $, where
$\bm{r}$ is represented by its spherical coordinates $(r,\theta,\psi)$.  It has
been shown that any harmonic function in a three-dimensional space can be
represented as a series expansion of spherical harmonic functions
$Y_{lm}(\theta,\phi)$:
\begin{equation} \label{eq:harmonic}
    V(\bm{r}) = \sum_{l=1}^{\infty} \sum_{m=-l}^{l} \alpha_{lm}
    \frac{Y_{lm}(\theta,\phi)}{r^{l+1}} + \sum_{l=1}^{\infty} \sum_{m=-l}^{l}
    \beta_{lm}
    r^{l} Y_{lm}(\theta,\phi)
\end{equation}

We can separate this expansion into two sets of functions: those proportional
to inverse powers of $r$ and those proportional to powers of $r$.
From a given array of sensors and a coordinate system with
its origin somewhere inside of the helmet, we can compute the
signal vectors corresponding to each of the terms in \ref{eq:harmonic}.

Following notations of \cite{taulu}, let $\bm{a}_{lm}$ be the signal vector corresponding to term
$\frac{Y_{lm}(\theta,\phi)}{r^{l+1}}$ and $\bm{b}_{lm}$ the signal
vector corresponding to $r^{l}Y_{lm}(\theta,\phi)$. A set of $P$ such
signal vectors forms a basis
in the $P$ dimensional signal space, and hence, the signal
vector is given as

\begin{equation}
    \bm{x} = \sum_{l=1}^{\infty} \sum_{m=-l}^{l} \alpha_{lm}
    \bm{a}_{lm} + \sum_{l=1}^{\infty} \sum_{m=-l}^{l} \beta_{lm}
    \bm{b}_{lm}
\end{equation}

This basis is not orthogonal, but linearly independent so any measured
signal vector has a unique representation in this basis:
\begin{equation}
        \bm{x} = \left[
                \bm{S_{in}} \quad \bm{S_{out}} \right]
            \begin{bmatrix}\bm{x}_{in}\\\bm{x}_{out}\end{bmatrix}
\end{equation}

where the sub-bases $ \bm{S}_{in}$ and $ \bm{S}_{out}$ contain
the basis vectors
$ \bm{a}_{lm}$ and $ \bm{b}_{lm}$, and vectors
$ \bm{x}_{in}$ and $ \bm{x}_{out}$ contain the corresponding
$\alpha_{lm}$ and $\beta_{lm}$ values.

It can be shown that the spherical harmonic functions contain
increasingly higher spatial frequencies when going to higher index values (l,m)
so that the signals from real magnetic sources are mostly contained in the low
$l,m$ end of the spectrum. By discarding the high $l,m$ end of the spectrum we
thus reduce the noise.  Then we can do signal space separation. It can be
shown that the basis vectors corresponding to the terms in the second sum in
expansion ~\eqref{eq:harmonic} represent the perturbating sources external to the helmet.
We can than separate the components of field arising from sources inside and
outside of the helmet. By discarding them we are left with the part of the
signal coming from inside of the helmet only.  The signal vector $\bm{x}$ is
then decomposed into 2 components $\bm{\phi}_{in}$ and $\bm{\phi}_{out}$ with
$\bm{\phi}_{in} = \bm{S}_{in} \bm{x}_{in}$ reproducing in all the MEG channels the signals that would be
seen if no interference from sources external to the helmet existed.


The real data from the Cam-CAN dataset have been measured with an Elekta
Neuromag 306-channel device, the only one that has been extensively tested on
Maxwell Filtering. For this device we included components up to
$l = L_{in} = 8$ for the $\bm{S}_{in}$ basis, and up to $l = L_{out} = 3$ for
the $ \bm{S}_{out}$ basis.

SSS requires a comprehensive sampling (more than about 150 channels) and a
relatively high calibration accuracy that is machine/site-specific.  For this
purpose we used the fine-calibration coefficients and the cross-talk correction
information provided in the Can-CAM repository for the 306-channels Neuromag
system used in this study.

For this study we used the temporal SSS (tSSS) extension~\cite{taulu},
where both temporal and spatial projection are applied to the MEG data. We used
an order 8 (resp. 3) of internal (resp. external) component of spherical
expansion, a 10s sliding window, a correlation threshold of 98\% (limit between
inner and outer subspaces used to reject overlapping intersecting inner/outer
signals), basis regularization, no movement compensation.

The origin of internal and external multipolar moment space is fitted via
head-digitization hence specified in the 'head' coordinate frame and the median
head position during the 10s window is used.

After projection in the lower-dimensional SSS basis we project back the signal
in its original space producing a signal $\bm{X}^{clean} = \bm{S}_{in}^{\top}
\bm{S}_{in} \bm{X} \in \bbR^{P \times T}$ with a much better SNR (reduced noise
variance) but with a rank $R \leq P$.
As a result each reconstructed sensor is then a linear
combination of $R$ synthetic source signals, which modifies the
inter-channel correlation structure, rendering the covariance matrix
rank-deficient.


\paragraph{Signal Space Projection (SSP)}

Recalling the MEG generative model \eqref{eq:generativemodel} if one knows, or can estimate, K linearly independent source patterns $\bm{a}_{1}, \ldots, \bm{a}_{K}$ that span the space $S = \textrm{span}(\bm{a}_{1}, \ldots, \bm{a}_{K}) \subset \bbR^P$ that contains the brain signal, one can estimate an orthonormal basis $U_{K} \in \bbR^{P \times K}$ of $S$ by singular value decomposition (SVD). One can then project any sensor space signal $\bm{x} \in \bbR^P$ onto $S$ to improve the SNR. The projection reads:
$$
    U_{K}U_{K}^{\top}\bm{x} \enspace .
$$

This is the idea behind the
Signal Space Projections (SSP) method \cite{uusitalo1997signal}.
In practice SSP is used to reduce physiological artifacts (eye blinks and heart beats) that cause prominent artifacts in the recording. In the Cam-CAN dataset eye blinks
are monitored by 2 electro-oculogram (EOG channels), and heart beats by an
electro-cardiogram (ECG channel).

SSP projections are computed from time segments contaminated by the artifacts
and the first component (per artifact and sensor type) are projected out.  More
precisely, the EOG and ECG channels are used to identify the artifact events
(after a first band-pass filter to remove DC offset and an additional [1-10]Hz
filter applied only to EOG channels to remove saccades vs blinks). After
filtering the raw signal in [1-35]Hz band, data segments (called epochs) are
created around those events, rejecting those whose peak-to-peak amplitude
exceeds a certain global threshold (see section below). For each artifact and
sensor type those epochs are then averaged and the first component of maximum
variance is extracted via PCA. Signal is then projected in the orthogonal space. This follows the guidelines of the MNE software~\cite{mne}.

\paragraph{Marking bad data segments}

We epoch the resulting data in 30s non overlapping windows and identify bad data segments (i.e. trials containing transient jumps in isolated channels) that have a peak-to-peak amplitude exceeding a certain global threshold, learnt automatically from the data using the autoreject (global) algorithm \cite{jas2017autoreject}.

%

\end{document}